\documentclass[a4paper,USenglish,cleveref,autoref,thm-restate]{lipics-v2021}

\pdfoutput=1
\hideLIPIcs

\title{Higher-Order Constrained Dependency Pairs for (Universal) Computability}

\author{Liye Guo}{Radboud University, Nijmegen, The Netherlands}{l.guo@cs.ru.nl}{https://orcid.org/0000-0002-3064-2691}{}

\author{Kasper Hagens}{Radboud University, Nijmegen, The Netherlands}{k.hagens@cs.ru.nl}{https://orcid.org/0009-0005-2382-0559}{}

\author{Cynthia Kop}{Radboud University, Nijmegen, The Netherlands}{c.kop@cs.ru.nl}{https://orcid.org/0000-0002-6337-2544}{}

\author{Deivid Vale}{Radboud University, Nijmegen, The Netherlands}{d.vale@cs.ru.nl}{https://orcid.org/0000-0003-1350-3478}{}

\authorrunning{L. Guo, K. Hagens, C. Kop and D. Vale}

\Copyright{Liye Guo, Kasper Hagens, Cynthia Kop and Deivid Vale}



\supplement{The source code is available on GitHub \cite{cora} and the version for this paper on Zenodo \cite{coraMFCS:24}.}

\funding{The authors are supported by NWO VI.Vidi.193.075, project ``CHORPE''.}

\nolinenumbers

\EventEditors{Rastislav Kr\'{a}lovi\v{c} and Anton\'{i}n Ku\v{c}era}
\EventNoEds{2}
\EventLongTitle{49th International Symposium on Mathematical Foundations of Computer Science (MFCS 2024)}
\EventShortTitle{MFCS 2024}
\EventAcronym{MFCS}
\EventYear{2024}
\EventDate{August 26--30, 2024}
\EventLocation{Bratislava, Slovakia}
\EventLogo{}
\SeriesVolume{306}
\ArticleNo{49}

\usepackage{mathtools,tikz}
\usetikzlibrary{positioning}

\newcommand{\eiLabel}[1]{(#1)}
\newenvironment{enumerate*}{\setcounter{enumi}{0}\renewcommand{\item}{\refstepcounter{enumi}\eiLabel{\theenumi}~}\ignorespaces}{\ignorespacesafterend}

\newcommand{\setBrac}[1]{\{\, #1 \,\}}
\newcommand{\setComp}[2]{\setBrac{#1 \,|\, #2}}
\newcommand{\setTheo}[1]{{#1}_\vartheta}
\newcommand{\setSort}{\mathcal{S}}
\newcommand{\setType}{\mathcal{T}}
\newcommand{\setFunc}{\mathcal{F}}
\newcommand{\setVari}{\mathcal{V}}
\newcommand{\setPret}{\mathfrak{T}}
\newcommand{\setTerm}[2]{T({#1}, {#2})}
\newcommand{\setFvar}[1]{\operatorname{Var}(#1)}
\newcommand{\setAlge}{\mathfrak{X}}
\newcommand{\setInte}{\mathbb{Z}}
\newcommand{\setRule}{\mathcal{R}}
\newcommand{\setDfnd}{\mathcal{D}}
\newcommand{\setAcce}[1]{\operatorname{Acc}(#1)}
\newcommand{\setReCa}{\mathbb{C}}
\newcommand{\setStDP}[1]{\operatorname{SDP}(#1)}
\newcommand{\setDePa}{\mathcal{P}}
\newcommand{\setHead}[1]{\operatorname{heads}(#1)}
\newcommand{\setFreI}[1]{\iota(#1)}
\newcommand{\setFreV}[1]{\mathcal{X}(#1)}

\newcommand{\arrType}{\rightarrow}
\newcommand{\arrRule}{\rightarrow}
\newcommand{\arrRwrt}{\arrRule}
\newcommand{\arrNorm}{\downarrow}
\newcommand{\arrComp}{\Rrightarrow}
\newcommand{\arrDePa}{\Rightarrow}

\newcommand{\relComp}[2]{{#1}\mathrel{;}{#2}}
\newcommand{\relType}{:}
\newcommand{\relSupt}{\trianglerighteq}
\newcommand{\relSupp}{\vartriangleright}
\newcommand{\relSord}{\succsim}
\newcommand{\relSoSt}{\succ}
\newcommand{\relSoRe}{\precsim}
\newcommand{\relSoPo}{\relSord_+}
\newcommand{\relSoNe}{\relSoSt_-}
\newcommand{\relAcce}{\relSupt_\mathrm{acc}}
\newcommand{\relThEq}{\equiv}

\newcommand{\mapNtrp}[1]{[\![#1]\!]}
\newcommand{\mapNorm}[1]{#1\,{\arrNorm}}

\newcommand{\lblCalc}{\kappa}

\newcommand{\flagAny}{\mathfrak{an}}
\newcommand{\flagPub}{\mathfrak{pu}}

\newcommand{\vect}[4]{{#1}_{#2} #4 {#1}_{#3}}
\newcommand{\appl}[2]{{#1}\ {#2}}
\newcommand{\appB}[3]{\appl{\appl{#1}{#2}}{#3}}
\newcommand{\appT}[4]{\appl{\appB{#1}{#2}{#3}}{#4}}
\newcommand{\bind}[3]{#1 {#2}.\, #3}
\newcommand{\sbst}[2]{{#1} \coloneq {#2}}
\newcommand{\hole}{\square}
\newcommand{\plug}[2]{{#1}[#2]}
\newcommand{\rwrl}[3]{{#1}\arrRule{#2}\ [#3]}
\newcommand{\shrp}[1]{{#1}^\sharp}
\newcommand{\dpso}{\mathsf{dp}}
\newcommand{\dpnc}[2]{{#1}\arrDePa{#2}}
\newcommand{\dpwc}[3]{\dpnc{#1}{#2}\ [#3]}
\newcommand{\dpcv}[4]{\dpwc{#1}{#2}{#3 \mid #4}}
\newcommand{\crel}[5]{{#2} \mathrel{#1}_{#4}^{#5} {#3}}
\newcommand{\pflg}{\mathtt{p}}

\newcommand{\typA}{A}
\newcommand{\typB}{B}
\newcommand{\typC}{C}
\newcommand{\trmA}{t}
\newcommand{\trmB}{s}
\newcommand{\trmC}{u}
\newcommand{\trmD}{v}
\newcommand{\funA}{f}
\newcommand{\funB}{g}
\newcommand{\varA}{x}
\newcommand{\varB}{y}
\newcommand{\varC}{z}
\newcommand{\subA}{\sigma}
\newcommand{\ctxA}{C}
\newcommand{\rulL}{\ell}
\newcommand{\rulR}{r}
\newcommand{\rulC}{\varphi}
\newcommand{\valA}{v}
\newcommand{\rcsA}{I}
\newcommand{\gtvA}{L}
\newcommand{\dppA}{\rho}
\newcommand{\grhA}{G}
\newcommand{\homA}{\theta}
\newcommand{\sdpA}{p}
\newcommand{\sccA}{S}
\newcommand{\prjA}{\nu}
\newcommand{\inmA}{\mathcal{J}}
\newcommand{\tamA}{\tau}
\newcommand{\relA}{R}

\newcommand{\mapf}{\mathsf{map}}
\newcommand{\fold}{\mathsf{fold}}
\newcommand{\lapp}{\mathsf{app}}
\newcommand{\llam}{\mathsf{lam}}
\newcommand{\fact}{\mathsf{fact}}
\newcommand{\comp}{\mathsf{comp}}
\newcommand{\iden}{\mathsf{id}}
\newcommand{\coml}{\mathsf{complst}}
\newcommand{\gcdl}{\mathsf{gcdlist}}
\newcommand{\gcdf}{\mathsf{gcd}}
\newcommand{\init}{\mathsf{init}}

\newcommand{\tyIn}{\mathsf{int}}
\newcommand{\tyBo}{\mathsf{bool}}
\newcommand{\tyLF}{\mathsf{funlist}}
\newcommand{\tyLI}{\mathsf{intlist}}
\newcommand{\booF}{\mathfrak{f}}
\newcommand{\booT}{\mathfrak{t}}
\newcommand{\algF}{0}
\newcommand{\algT}{1}
\newcommand{\fnlN}{\mathsf{fnil}}
\newcommand{\fnlC}{\mathsf{fcons}}
\newcommand{\lstN}{\mathsf{nil}}
\newcommand{\lstC}{\mathsf{cons}}

\newcommand{\cora}{\textsf{Cora}}
\newcommand{\tcap}{\zeta}

\begin{document}

\maketitle

\begin{abstract}
  Dependency pairs constitute a series of very effective techniques for
  the termination analysis of term rewriting systems.
  In this paper, we adapt the static dependency pair framework to
  logically constrained simply-typed term rewriting systems (LCSTRSs),
  a higher-order formalism with logical constraints built in.
  We also propose the concept of universal computability,
  which enables a form of open-world termination analysis
  through the use of static dependency pairs.
\end{abstract}

\section{Introduction}\label{sec:introduction}
Logically constrained simply-typed term rewriting systems (LCSTRSs) \cite{guo:kop:24}
are a formalism of higher-order term rewriting with logical constraints
(built on its first-order counterpart \cite{kop:nis:13}).
Proposed for program analysis,
LCSTRSs offer a flexible representation of programs
owing to---in contrast with traditional TRSs---%
their support for primitive data types such as
(arbitrary-precision or fixed-width) integers and floating-point numbers.
Without compromising the ability
to directly reason about these widely used data types,
LCSTRSs bridge the gap between
the abundant techniques based on term rewriting and
automatic program analysis.

We consider \emph{termination} analysis in this paper.
The termination of LCSTRSs was first discussed in \cite{guo:kop:24}
through a variant of the higher-order recursive path ordering (HORPO) \cite{jou:rub:99}.
This paper furthers that discussion
by introducing dependency pairs \cite{art:gie:00} to LCSTRSs.
As a broad framework for termination,
this method was initially proposed for unconstrained first-order term rewriting,
and was later generalized in a variety of higher-order settings
(see, e.g., \cite{sak:wat:sak:01,kop:raa:12,sak:kus:05,bla:06}).
Modern termination analyzers rely heavily on dependency pairs.

In higher-order termination analysis, dependency pairs take two forms:\ %
the dynamic \cite{sak:wat:sak:01,kop:raa:12} and
the static \cite{sak:kus:05,bla:06,kus:sak:07,fuh:kop:19}.
This paper concentrates on \emph{static} dependency pairs,
and is based on the definitions in \cite{fuh:kop:19,kus:sak:07}.
First-order dependency pairs with logical constraints
have been informally defined by the third author \cite{kop:13},
from which we also take inspiration.

For program analysis,
the traditional notion of termination can be inefficient,
and arguably insufficient.
It assumes that the whole program is known and analyzed,
i.e., \emph{closed-world} analysis.
This way even small programs that happen to import a large standard library
need sophisticated analysis;
local changes in a multipart, previously verified program
also require the entire analysis to be redone.
As O'Hearn \cite{hea:18} argues (in a different context),
studying \emph{open-world} analysis opens up many applications.
In particular,
it is practically desirable to analyze the termination of standard libraries---%
or modules of a larger program in general---%
without prior knowledge of how the functions they define may be used.

It is tricky to characterize such a property,
especially in the presence of higher-order arguments.
For example, \( \mapf \) and \( \fold \) are usually considered ``terminating'',
even though passing a non-terminating function to them can surely result in non-termination.
Hence, we need to narrow our focus to certain ``reasonable'' calls.
On the other hand, the function
\( \appl{\lapp}{(\appl{\llam}{f})} \arrRule f \) where
\( \lapp \relType \mathsf{o} \arrType \mathsf{o} \arrType \mathsf{o} \) and
\( \llam \relType (\mathsf{o} \arrType \mathsf{o}) \arrType \mathsf{o} \)
would generally be considered ``non-terminating'', because
if we define \( \appl{\mathsf{w}}{x} \arrRule \appB{\lapp}{x}{x} \),
an infinite rewrite sequence starts from
\( \appB{\lapp}{(\appl{\llam}{\mathsf{w}})}{(\appl{\llam}{\mathsf{w}})} \)---%
this encodes the famous \( \Omega \) in the untyped \( \lambda \)-calculus.
The property we are looking for must distinguish \( \mapf \) and \( \fold \) from \( \lapp \).

To capture this property,
we propose a new concept, called \emph{universal computability}.
In light of information hiding,
this concept can be further generalized to \emph{public computability}.
We will see that static dependency pairs
are a natural vehicle for analyzing these properties.

Various modular aspects of term rewriting have been studied by the community.
Our scenario roughly corresponds to hierarchical combinations
\cite{rao:93,rao:94,rao:95,der:95},
where different parts of a program are analyzed separately.
We follow this terminology so that
it will be easier to compare our work with the literature.
However, our setup---higher-order constrained rewriting---%
is separate from the first-order and unconstrained setting
in which hierarchical combinations were initially proposed.
Furthermore, our approach has a different focus---%
namely, the use of static dependency pairs.

\subparagraph*{Contributions.}
We recall the formalism of LCSTRSs and the predicate of computability in
\cref{sec:preliminaries}.
Then the contributions of this paper follow:
\begin{itemize}
\item We propose in \cref{sec:dependencypairs} the first definition of
  \emph{dependency pairs} for higher-order logically constrained TRSs.
  This is also the first DP approach for constrained rewriting
  as the prior work on first-order constrained dependency pairs \cite{kop:13}
  has never been formally published.
\item We define in \cref{sec:dpframework} the \emph{constrained DP framework}
  for termination analysis with five classes of \emph{DP processors},
  which can be used to simplify termination problems.
\item We extend the notion of a \emph{hierarchical combination}
  \cite{rao:93,rao:94,rao:95,der:95} to LCSTRSs and
  define \emph{universal} and \emph{public computability} in \cref{sec:computability}.
  We also fine-tune the DP framework to support these properties,
  and provide extra DP processors for public computability.
  This allows the DP framework to be used for open-world analysis.
  We base \cref{sec:computability} on LCSTRSs
  to emphasize those notions in real-world programming,
  but they are new and of theoretical interest in higher-order term rewriting
  even without logical constraints.
\item We have implemented the DP framework for both termination and public computability
  in our open-source analyzer \cora.
  We describe the experimental evaluation in \cref{sec:experiments}.
\end{itemize}

\section{Preliminaries}\label{sec:preliminaries}
In this section, we collect
the preliminary definitions and results we need
from the literature.
First, we recall the definition of an LCSTRS \cite{guo:kop:24}.
In this paper, we put a restriction on rewrite rules:\ %
\( \rulL \) is always a pattern in \( \rwrl{\rulL}{\rulR}{\rulC} \).
Next, we recall the definition of computability (with accessibility) from \cite{fuh:kop:19}.
This version is particularly tailored for static dependency pairs.

\subsection{Logically Constrained STRSs}

\subparagraph*{Terms Modulo Theories.}
Given a non-empty set \( \setSort \) of \emph{sorts} (or \emph{base types}),
the set \( \setType \) of simple types over \( \setSort \) is generated by
the grammar \( \setType \Coloneqq \setSort \mid (\setType \arrType \setType) \).
Right-associativity is assigned to \( \arrType \) so we can omit some parentheses.
Given disjoint sets \( \setFunc \) and \( \setVari \),
whose elements we call \emph{function symbols} and \emph{variables}, respectively,
the set \( \setPret \) of \emph{pre-terms} over \( \setFunc \) and \( \setVari \)
is generated by the grammar
\( \setPret \Coloneqq \setFunc \mid \setVari \mid (\appl{\setPret}{\setPret}) \).
Left-associativity is assigned to the juxtaposition operation, called \emph{application},
so \( \appB{\trmA_0}{\trmA_1}{\trmA_2} \) stands for
\( (\appl{(\appl{\trmA_0}{\trmA_1})}{\trmA_2}) \), for example.

We assume every function symbol and variable is assigned a unique type.
Typing works as expected:\ %
if pre-terms \( \trmA_0 \) and \( \trmA_1 \) have
types \( \typA \arrType \typB \) and \( \typA \), respectively,
\( \appl{\trmA_0}{\trmA_1} \) has type \( \typB \).
The set \( \setTerm{\setFunc}{\setVari} \) of \emph{terms} over \( \setFunc \) and \( \setVari \)
consists of pre-terms that have a type.
We write \( \trmA \relType \typA \) if a term \( \trmA \) has type \( \typA \).
We assume there are infinitely many variables of each type.

The set \( \setFvar{\trmA} \) of variables in \( \trmA \in \setTerm{\setFunc}{\setVari} \)
is defined by
\( \setFvar{\funA} = \emptyset \) for \( \funA \in \setFunc \),
\( \setFvar{\varA} = \setBrac{\varA} \) for \( \varA \in \setVari \) and
\( \setFvar{\appl{\trmA_0}{\trmA_1}} = \setFvar{\trmA_0} \cup \setFvar{\trmA_1} \).
A term \( \trmA \) is called \emph{ground}
if \( \setFvar{\trmA} = \emptyset \).

For constrained rewriting,
we make further assumptions.
First, we assume
that there is a distinguished subset \( \setTheo{\setSort} \) of \( \setSort \),
called the set of \emph{theory sorts}.
The grammar
\( \setTheo{\setType} \Coloneqq \setTheo{\setSort} \mid (\setTheo{\setSort} \arrType \setTheo{\setType}) \)
generates the set \( \setTheo{\setType} \) of
\emph{theory types} over \( \setTheo{\setSort} \).
Note that a theory type is essentially a non-empty list of theory sorts.
Next, we assume
that there is a distinguished subset \( \setTheo{\setFunc} \) of \( \setFunc \),
called the set of \emph{theory symbols},
and that the type of every theory symbol is in \( \setTheo{\setType} \),
which means that the type of any argument passed to a theory symbol is a theory sort.
Theory symbols whose type is a theory sort are called \emph{values}.
Elements of \( \setTerm{\setTheo{\setFunc}}{\setVari} \) are called \emph{theory terms}.

Theory symbols are interpreted in an underlying theory:\ %
given an \( \setTheo{\setSort} \)-indexed family of sets
\( (\setAlge_\typA)_{\typA \in \setTheo{\setSort}} \),
we extend it to a \( \setTheo{\setType} \)-indexed family by
letting \( \setAlge_{\typA \arrType \typB} \) be the set of mappings
from \( \setAlge_\typA \) to \( \setAlge_\typB \);
an \emph{interpretation} of theory symbols
is a \( \setTheo{\setType} \)-indexed family of mappings
\( (\mapNtrp{\cdot}_\typA)_{\typA \in \setTheo{\setType}} \) where
\( \mapNtrp{\cdot}_\typA \) assigns to each theory symbol of type \( \typA \)
an element of \( \setAlge_\typA \) and is bijective
if \( \typA \in \setTheo{\setSort} \).
Given an interpretation of theory symbols \( (\mapNtrp{\cdot}_\typA)_{\typA \in \setTheo{\setType}} \),
we extend each indexed mapping \( \mapNtrp{\cdot}_\typB \) to one
that assigns to each \emph{ground theory term} of type \( \typB \)
an element of \( \setAlge_\typB \) by
letting \( \mapNtrp{\appl{\trmA_0}{\trmA_1}}_\typB \) be
\( \mapNtrp{\trmA_0}_{\typA \arrType \typB}(\mapNtrp{\trmA_1}_\typA) \).
We write just \( \mapNtrp{\cdot} \) when the type can be deduced.

\begin{example}
  Let \( \setTheo{\setSort} \) be \( \setBrac{\tyIn} \).
  Then \( \tyIn \arrType \tyIn \arrType \tyIn \)
  is a theory type over \( \setTheo{\setSort} \)
  while \( (\tyIn \arrType \tyIn) \arrType \tyIn \) is not.
  Let \( \setTheo{\setFunc} \) be
  \( \setBrac{-} \cup \setInte \) where
  \( {-} \relType \tyIn \arrType \tyIn \arrType \tyIn \) and
  \( n \relType \tyIn \) for all \( n \in \setInte \).
  The values are the elements of \( \setInte \).
  Let \( \setAlge_\tyIn \) be \( \setInte \),
  \( \mapNtrp{\cdot}_\tyIn \) be the identity mapping and
  \( \mapNtrp{-} \) be the mapping \( \bind{\lambda}{m}{\bind{\lambda}{n}{m - n}} \).
  The interpretation of \( \appl{(-)}{1} \)
  is the mapping \( \bind{\lambda}{n}{1 - n} \).
\end{example}

\subparagraph*{Substitutions, Contexts and Subterms.}
Type-preserving mappings
from \( \setVari \) to \( \setTerm{\setFunc}{\setVari} \) are called \emph{substitutions}.
Every substitution \( \subA \) extends to a type-preserving mapping \( \bar{\subA} \)
from \( \setTerm{\setFunc}{\setVari} \) to \( \setTerm{\setFunc}{\setVari} \).
We write \( \trmA \subA \) for \( \bar{\subA}(\trmA) \) and define it as follows:\ %
\( \funA \subA = \funA \) for \( \funA \in \setFunc \),
\( \varA \subA = \subA(\varA) \) for \( \varA \in \setVari \) and
\( (\appl{\trmA_0}{\trmA_1}) \subA = \appl{(\trmA_0 \subA)}{(\trmA_1 \subA)} \).
Let \( [\sbst{\varA_1}{\trmA_1}, \ldots, \sbst{\varA_n}{\trmA_n}] \) denote
the substitution \( \subA \) such that
\( \subA(\varA_i) = \trmA_i \)
for all \( i \),
and \( \subA(\varB) = \varB \)
for all \( \varB \in \setVari \setminus \setBrac{\vect{\varA}{1}{n}{, \ldots,}} \).

A context is a term containing a hole.
Let \( \hole \) be a special terminal symbol and
assign to it a type \( \typA \);
a \emph{context} \( \plug{\ctxA}{} \)
is an element of \( \setTerm{\setFunc}{\setVari \cup \setBrac{\hole}} \)
such that \( \hole \) occurs in \( \plug{\ctxA}{} \) exactly once.
Given a term \( \trmA \relType \typA \),
let \( \plug{\ctxA}{\trmA} \) denote the term produced by
replacing \( \hole \) in \( \plug{\ctxA}{} \) with \( \trmA \).

A term \( \trmA \) is called a (maximally applied) \emph{subterm} of a term \( \trmB \),
written as \( \trmB \relSupt \trmA \),
if either \( \trmB = \trmA \),
\( \trmB = \appl{\trmB_0}{\trmB_1} \)
where \( \trmB_1 \relSupt \trmA \), or
\( \trmB = \appl{\trmB_0}{\trmB_1} \)
where \( \trmB_0 \relSupt \trmA \) and \( \trmB_0 \neq \trmA \);
i.e., \( \trmB = \plug{\ctxA}{\trmA} \) for \( \plug{\ctxA}{} \)
that is not of form \( \plug{\ctxA^\prime}{\appl{\hole}{\trmA_1}} \).
We write \( \trmB \relSupp \trmA \)
and call \( \trmA \) a \emph{proper subterm} of \( \trmB \)
if \( \trmB \relSupt \trmA \) and \( \trmB \neq \trmA \).

\subparagraph*{Constrained Rewriting.}
Constrained rewriting requires the theory sort \( \tyBo \):\ %
we henceforth assume that
\( \tyBo \in \setTheo{\setSort} \),
\( \setBrac{\booF, \booT} \subseteq \setTheo{\setFunc} \),
\( \setAlge_\tyBo = \setBrac{\algF, \algT} \),
\( \mapNtrp{\booF}_\tyBo = \algF \) and
\( \mapNtrp{\booT}_\tyBo = \algT \).
A \emph{logical constraint} is a theory term \( \rulC \) such that
\( \rulC \) has type \( \tyBo \) and
the type of each variable in \( \setFvar{\rulC} \) is a theory sort.
A (constrained) \emph{rewrite rule} is a triple \( \rwrl{\rulL}{\rulR}{\rulC} \) where
\( \rulL \) and \( \rulR \) are terms which have the same type,
\( \rulC \) is a logical constraint,
the type of each variable in \( \setFvar{\rulR} \setminus \setFvar{\rulL} \)
is a theory sort and
\( \rulL \) is a pattern that
takes the form \( \appl{\funA}{\vect{\trmA}{1}{n}{\cdots}} \) for
some function symbol \( \funA \) and
contains at least one function symbol in \( \setFunc \setminus \setTheo{\setFunc} \).
Here a \emph{pattern} is a term whose subterms are either
\( \appl{\funA}{\vect{\trmA}{1}{n}{\cdots}} \) for some function symbol \( \funA \) or
a variable.
A substitution \( \subA \) is said
to \emph{respect} \( \rwrl{\rulL}{\rulR}{\rulC} \) if
\( \subA(\varA) \) is a value for all
\( \varA \in \setFvar{\rulC} \cup (\setFvar{\rulR} \setminus \setFvar{\rulL}) \) and
\( \mapNtrp{\rulC \subA} = \algT \).

A \emph{logically constrained simply-typed term rewriting system} (LCSTRS)
collects the above data---%
\( \setSort \), \( \setTheo{\setSort} \), \( \setFunc \), \( \setTheo{\setFunc} \),
\( \setVari \), \( (\setAlge_\typA) \) and \( \mapNtrp{\cdot} \)---%
along with a set \( \setRule \) of rewrite rules.
We usually let \( \setRule \) alone stand for the system.
The set \( \setRule \) induces the \emph{rewrite relation}
\( {\arrRwrt_\setRule} \) over terms:\ %
\( \trmA \arrRwrt_\setRule \trmA^\prime \) if and only if
there exists a context \( \plug{\ctxA}{} \) such that either
\begin{enumerate*}
\item \( \trmA = \plug{\ctxA}{\rulL \subA} \) and
  \( \trmA^\prime = \plug{\ctxA}{\rulR \subA} \) for
  some rewrite rule \( \rwrl{\rulL}{\rulR}{\rulC} \in \setRule \) and
  some substitution \( \subA \) which respects \( \rwrl{\rulL}{\rulR}{\rulC} \), or
\item \( \trmA = \plug{\ctxA}{\appl{\funA}{\vect{\valA}{1}{n}{\cdots}}} \) and
  \( \trmA^\prime = \plug{\ctxA}{\valA^\prime} \) for
  some theory symbol \( \funA \) and
  some values \( \vect{\valA}{1}{n}{, \ldots,}, \valA^\prime \) with
  \( n > 0 \) and
  \( \mapNtrp{\appl{\funA}{\vect{\valA}{1}{n}{\cdots}}} = \mapNtrp{\valA^\prime} \).
\end{enumerate*}
When no ambiguity arises,
we may write \( \arrRwrt \) for \( \arrRwrt_\setRule \).

If \( \trmA \arrRwrt_\setRule \trmA^\prime \) due to the second condition above,
we also write \( \trmA \arrRwrt_\lblCalc \trmA^\prime \) and
call it a \emph{calculation step}.
Theory symbols that are not a value
are called \emph{calculation symbols}.
Let \( \mapNorm{\trmA}_\lblCalc \) denote the (unique) \( \lblCalc \)-normal form of \( \trmA \),
i.e., the term \( \trmA^\prime \) such that
\( \trmA \arrRwrt_\lblCalc^* \trmA^\prime \) and
\( \trmA^\prime \not\arrRwrt_\lblCalc \trmA^{\prime\prime} \) for any \( \trmA^{\prime\prime} \).
For example,
\( \mapNorm{(\appl{\funA}{(7 * (3 * 2))})}_\lblCalc = \appl{\funA}{42} \) if
\( \funA \) is not a calculation symbol, or if
\( \funA \relType \tyIn \arrType \typA \arrType \typB \).

A rewrite rule \( \rwrl{\rulL}{\rulR}{\rulC} \)
is said to \emph{define} a function symbol \( \funA \) if
\( \rulL = \appl{\funA}{\vect{\trmA}{1}{n}{\cdots}} \).
Given an LCSTRS \( \setRule \),
\( \funA \) is called a \emph{defined symbol} if
some rewrite rule in \( \setRule \) defines \( \funA \).
Let \( \setDfnd \) denote the set of defined symbols.
Values and function symbols in \( \setFunc \setminus (\setTheo{\setFunc} \cup \setDfnd) \)
are called \emph{constructors}.

\begin{example}\label{ex:fact}
  Below is the factorial function in continuation-passing style as an LCSTRS:
  \begin{align*}
    \appB{\fact}{n}{k}    &\arrRule \appl{k}{1}                                               &&[n \le 0] &
    \appT{\comp}{g}{f}{x} &\arrRule \appl{g}{(\appl{f}{x})}\\
    \appB{\fact}{n}{k}    &\arrRule \appB{\fact}{(n - 1)}{(\appB{\comp}{k}{(\appl{(*)}{n})})} &&[n > 0]   &
    \appl{\iden}{x}       &\arrRule x
  \end{align*}
  We use infix notation for some binary operators,
  and omit the logical constraint of a rewrite rule when it is \( \booT \).
  An example rewrite sequence is
  \( \appB{\fact}{1}{\iden}
     \arrRwrt
     \appB{\fact}{(1 - 1)}{(\appB{\comp}{\iden}{(\appl{(*)}{1})})}
     \arrRwrt_\lblCalc
     \appB{\fact}{0}{(\appB{\comp}{\iden}{(\appl{(*)}{1})})}
     \arrRwrt
     \appT{\comp}{\iden}{(\appl{(*)}{1})}{1}
     \arrRwrt
     \appl{\iden}{(\appB{(*)}{1}{1})}
     \arrRwrt_\lblCalc
     \appl{\iden}{1}
     \arrRwrt
     1 \).
\end{example}

\subsection{Accessibility and Computability}
We recall the notion of computability with accessibility---%
which originates from \cite{bla:jou:oka:02} and
is reformulated in \cite{fuh:kop:19} to couple with static dependency pairs---%
and adapt the notion of accessible function passing \cite{fuh:kop:19} to LCSTRSs.

\subparagraph*{Accessibility.}
Assume given a \emph{sort ordering}---%
a quasi-ordering \( \relSord \) over \( \setSort \)
whose strict part \( {\relSoSt} = {\relSord} \setminus {\relSoRe} \) is well-founded.
We inductively define two relations
\( \relSoPo \) and \( \relSoNe \) over \( \setSort \) and \( \setType \):\ %
given a sort \( \typA \) and
a type \( \typB = \vect{\typB}{1}{n}{\arrType \cdots \arrType} \arrType \typC \)
where \( \typC \) is a sort and \( n \ge 0 \),
\( \typA \relSoPo \typB \) if and only if
\( \typA \relSord \typC \) and \( \bind{\forall}{i}{\typA \relSoNe \typB_i} \), and
\( \typA \relSoNe \typB \) if and only if
\( \typA \relSoSt \typC \) and \( \bind{\forall}{i}{\typA \relSoPo \typB_i} \).

Given a function symbol
\( \funA \relType \vect{\typA}{1}{n}{\arrType \cdots \arrType} \arrType \typB \)
where \( \typB \) is a sort,
the set \( \setAcce{\funA} \) of the \emph{accessible argument positions} of \( \funA \)
is defined as \( \setComp{1 \le i \le n}{\typB \relSoPo \typA_i} \).
A term \( \trmA \) is called an \emph{accessible subterm} of a term \( \trmB \),
written as \( \trmB \relAcce \trmA \),
if either \( \trmB = \trmA \), or
\( \trmB = \appl{\funA}{\vect{\trmB}{1}{m}{\cdots}} \) for some \( \funA \in \setFunc \) and
there exists \( k \in \setAcce{\funA} \) such that \( \trmB_k \relAcce \trmA \).
An LCSTRS \( \setRule \) is called \emph{accessible function passing} (AFP) if
there exists a sort ordering such that
for all \( \rwrl{\appl{\funA}{\vect{\trmB}{1}{m}{\cdots}}}{\rulR}{\rulC} \in \setRule \) and
\( \varA \in \setFvar{\appl{\funA}{\vect{\trmB}{1}{m}{\cdots}}} \cap \setFvar{\rulR} \setminus \setFvar{\rulC} \),
there exists \( k \) such that \( \trmB_k \relAcce \varA \).

\begin{example}
  An LCSTRS \( \setRule \) is AFP (with \( \relSord \) equating all the sorts) if
  for all \( \rwrl{\appl{\funA}{\vect{\trmB}{1}{m}{\cdots}}}{\rulR}{\rulC} \in \setRule \)
  and \( i \in \setBrac{1, \ldots, m} \),
  the type of each proper subterm of \( \trmB_i \) is a sort.
  Rewrite rules for common higher-order functions, e.g., \( \mapf \) and \( \fold \),
  usually fit this criterion.

  Consider \( \setBrac{\appB{\coml}{\fnlN}{x} \arrRule x, \appB{\coml}{(\appB{\fnlC}{f}{l})}{x} \arrRule \appB{\coml}{l}{(\appl{f}{x})}} \),
  where \( \coml \relType \tyLF \arrType \tyIn \arrType \tyIn \)
  composes a list of \emph{functions}.
  This system is AFP with \( \tyLF \relSoSt \tyIn \).

  The system \( \setBrac{\appl{\lapp}{(\appl{\llam}{f})} \arrRule f} \)
  in \cref{sec:introduction}
  is not AFP since \( \mathsf{o} \relSoSt \mathsf{o} \) cannot be true.
\end{example}

\subparagraph*{Computability.}
A term is called \emph{neutral} if it takes the form
\( \appl{\varA}{\vect{\trmA}{1}{n}{\cdots}} \) for some variable \( \varA \).
A set of \emph{reducibility candidates}, or an \emph{RC-set},
for the rewrite relation \( \arrRwrt_\setRule \) of an LCSTRS \( \setRule \)
is an \( \setSort \)-indexed family of sets \( (\rcsA_\typA)_{\typA \in \setSort} \)
(let \( \rcsA \) denote \( \bigcup_\typA \rcsA_\typA \))
satisfying the following conditions:
\begin{bracketenumerate}
\item Each element of \( \rcsA_\typA \) is a terminating
  (with respect to \( \arrRwrt_\setRule \)) term of type \( \typA \).
\item Given terms \( \trmB \) and \( \trmA \) such that \( \trmB \arrRwrt_\setRule \trmA \),
  if \( \trmB \) is in \( \rcsA_\typA \), so is \( \trmA \).
\item Given a neutral term \( \trmB \),
  if \( \trmA \) is in \( \rcsA_\typA \) for all \( \trmA \) such that
  \( \trmB \arrRwrt_\setRule \trmA \),
  so is \( \trmB \).
\end{bracketenumerate}
Given an RC-set \( \rcsA \) for \( \arrRwrt_\setRule \),
a term \( \trmA_0 \) is called \emph{\( \rcsA \)-computable} if either
the type of \( \trmA_0 \) is a sort and \( \trmA_0 \in \rcsA \), or
the type of \( \trmA_0 \) is \( \typA \arrType \typB \) and
\( \appl{\trmA_0}{\trmA_1} \) is \( \rcsA \)-computable
for all \( \rcsA \)-computable \( \trmA_1 \relType \typA \).

We are interested in a specific RC-set \( \setReCa \),
whose existence is guaranteed by \cref{thm:existsC}.

\begin{theorem}[\textnormal{see \cite{fuh:kop:19}}]\label{thm:existsC}
  Given a sort ordering and an RC-set \( \rcsA \) for \( \arrRwrt_\setRule \),
  let \( \arrComp_\rcsA \) be the relation over terms such that
  \( \trmB \arrComp_\rcsA \trmA \) if and only if
  both \( \trmB \) and \( \trmA \) have a base type,
  \( \trmB = \appl{\funA}{\vect{\trmB}{1}{m}{\cdots}} \)
  for some function symbol \( \funA \),
  \( \trmA = \appl{\trmB_k}{\vect{\trmA}{1}{n}{\cdots}} \)
  for some \( k \in \setAcce{\funA} \) and
  \( \trmA_i \) is \( \rcsA \)-computable for all \( i \).

  Given an LCSTRS \( \setRule \) with a sort ordering,
  there exists an RC-set \( \setReCa \) for \( \arrRwrt_\setRule \) such that
  \( \trmA \in \setReCa_\typA \) if and only if
  \( \trmA \relType \typA \) is terminating with respect to
  \( {\arrRwrt_\setRule} \cup {\arrComp_\setReCa} \),
  and for all \( \trmA^\prime \) such that \( \trmA \arrRwrt_\setRule^* \trmA^\prime \),
  if \( \trmA^\prime = \appl{\funA}{\vect{\trmA}{1}{n}{\cdots}} \)
  for some function symbol \( \funA \),
  \( \trmA_i \) is \( \setReCa \)-computable for all \( i \in \setAcce{\funA} \).
\end{theorem}
Thus, given a \( \setReCa \)-computable term \( \appl{\funA}{\vect{\trmA}{1}{n}{\cdots}} \),
all its reducts and the accessible arguments---%
\( \trmA_i \) for \( i \in \setAcce{\funA} \)---%
are also \( \setReCa \)-computable.
We consider \( \setReCa \)-computability throughout this paper.

\section{Static Dependency Pairs for LCSTRSs}\label{sec:dependencypairs}
Originally proposed for unconstrained first-order term rewriting,
the dependency pair approach \cite{art:gie:00}---%
a methodology that analyzes the recursive structure of function calls---%
is at the heart of most modern automatic termination analyzers
for various styles of term rewriting.
There follow multiple higher-order generalizations,
among which we adopt here the \emph{static} branch \cite{kus:sak:07,fuh:kop:19}.
As we shall see in \cref{sec:computability},
this approach extends well to open-world analysis.

In this section,
we adapt static dependency pairs to LCSTRSs.
We start with a notation:

\begin{definition}
  Given an LCSTRS \( \setRule \),
  let \( \shrp{\setFunc} \) be
  \( \setFunc \cup \setComp{\shrp{\funA}}{\funA \in \setDfnd} \)
  where \( \setDfnd \) is the set of defined symbols in \( \setRule \) and
  \( \shrp{\funA} \) is a fresh function symbol for all \( \funA \).
  Let \( \dpso \) be a fresh sort,
  and for each defined symbol
  \( \funA \relType \vect{\typA}{1}{n}{\arrType \cdots \arrType} \arrType \typB \)
  where \( \typB \in \setSort \),
  we assign
  \( \shrp{\funA} \relType \vect{\typA}{1}{n}{\arrType \cdots \arrType} \arrType \dpso \).
  Given a term
  \( \trmA = \appl{\funA}{\vect{\trmA}{1}{n}{\cdots}} \in \setTerm{\setFunc}{\setVari} \)
  where \( \funA \in \setDfnd \),
  let \( \shrp{\trmA} \) denote
  \( \appl{\shrp{\funA}}{\vect{\trmA}{1}{n}{\cdots}} \in \setTerm{\shrp{\setFunc}}{\setVari} \).
\end{definition}

In the presence of logical constraints,
a dependency pair should be more than a pair.
Two extra components---a logical constraint and a set of variables---%
keep track of what substitutions are expected by the dependency pair.

\begin{definition}
  A \emph{static dependency pair} (SDP) is a quadruple
  \( \dpcv{\shrp{\trmB}}{\shrp{\trmA}}{\rulC}{\gtvA} \) where
  \( \shrp{\trmB} \) and \( \shrp{\trmA} \) are terms of type \( \dpso \),
  \( \rulC \) is a logical constraint and
  \( \gtvA \supseteq \setFvar{\rulC} \) is a set of variables
  whose types are theory sorts.
  Given a rewrite rule \( \rwrl{\rulL}{\rulR}{\rulC} \),
  let \( \setStDP{\rwrl{\rulL}{\rulR}{\rulC}} \) denote
  the set of SDPs of form
  \( \dpcv{\appl{\shrp{\rulL}}{\vect{\varA}{1}{m}{\cdots}}}{\appB{\shrp{\funB}}{\vect{\trmA}{1}{q}{\cdots}}{\vect{\varB}{q + 1}{n}{\cdots}}}{\rulC}{\setFvar{\rulC} \cup (\setFvar{\rulR} \setminus \setFvar{\rulL})} \)
  such that
  \begin{bracketenumerate}
  \item \( \shrp{\rulL} \relType \vect{\typA}{1}{m}{\arrType \cdots \arrType} \arrType \dpso \)
    while \( \varA_i \relType \typA_i \) is a fresh variable for all \( i \),
  \item \( \appl{\rulR}{\vect{\varA}{1}{m}{\cdots}} \relSupt \appl{\funB}{\vect{\trmA}{1}{q}{\cdots}} \)
    for \( \funB \in \setDfnd \), and
  \item \( \shrp{\funB} \relType \vect{\typB}{1}{n}{\arrType \cdots \arrType} \arrType \dpso \)
    while \( \varB_i \relType \typB_i \) is a fresh variable for all \( i > q \).
  \end{bracketenumerate}
  Let \( \setStDP{\setRule} \) be
  \( \bigcup_{\rwrl{\rulL}{\rulR}{\rulC} \in \setRule} \setStDP{\rwrl{\rulL}{\rulR}{\rulC}} \).
  A substitution \( \subA \) is said to \emph{respect} an SDP
  \( \dpcv{\shrp{\trmB}}{\shrp{\trmA}}{\rulC}{\gtvA} \) if
  \( \subA(\varA) \) is a ground theory term for all \( \varA \in \gtvA \) and
  \( \mapNtrp{\rulC \subA} = \algT \).
\end{definition}
The component \( \gtvA \) is new compared to \cite{kop:13}.
We shall see its usefulness in \cref{subsec:theoryarg},
as it gives us more freedom to manipulate dependency pairs.
We introduce two shorthand notations for SDPs:\ %
\( \dpwc{\shrp{\trmB}}{\shrp{\trmA}}{\rulC} \) for
\( \dpcv{\shrp{\trmB}}{\shrp{\trmA}}{\rulC}{\setFvar{\rulC}} \),
and \( \dpnc{\shrp{\trmB}}{\shrp{\trmA}} \) for
\( \dpcv{\shrp{\trmB}}{\shrp{\trmA}}{\booT}{\emptyset} \).

\begin{example}\label{ex:gcdlist}
  Consider the system \( \setRule \) consisting of the following rewrite rules, in which
  \( \gcdl \relType \tyLI \arrType \tyIn \),
  \( \fold \relType (\tyIn \arrType \tyIn \arrType \tyIn) \arrType \tyIn \arrType \tyLI \arrType \tyIn \) and
  \( \gcdf \relType \tyIn \arrType \tyIn \arrType \tyIn \).
  \begin{gather*}
    \gcdl \arrRule \appB{\fold}{\gcdf}{0}
    \qquad
    \appT{\fold}{f}{y}{\lstN} \arrRule y
    \qquad
    \appT{\fold}{f}{y}{(\appB{\lstC}{x}{l})} \arrRule \appB{f}{x}{(\appT{\fold}{f}{y}{l})}\\
    \begin{aligned}
      \appB{\gcdf}{m}{n} &\arrRule \appB{\gcdf}{(- m)}{n}       &&[m < 0] &
      \appB{\gcdf}{m}{n} &\arrRule \appB{\gcdf}{m}{(- n)}       &&[n < 0]\\
      \appB{\gcdf}{m}{0} &\arrRule m                            &&[m \ge 0] &
      \appB{\gcdf}{m}{n} &\arrRule \appB{\gcdf}{n}{(m \bmod n)} &&[m \ge 0 \wedge n > 0]
    \end{aligned}
  \end{gather*}
  The set \( \setStDP{\setRule} \) consists of
  \begin{enumerate*}
  \item\label{itm:gcdlistGcdlGcdf}%
    \( \dpnc{\appl{\shrp{\gcdl}}{l^\prime}}{\appB{\shrp{\gcdf}}{m^\prime}{n^\prime}} \),
  \item\label{itm:gcdlistGcdlFold}%
    \( \dpnc{\appl{\shrp{\gcdl}}{l^\prime}}{\appT{\shrp{\fold}}{\gcdf}{0}{l^\prime}} \),
  \item\label{itm:gcdlistFoldFold}%
    \( \dpnc{\appT{\shrp{\fold}}{f}{y}{(\appB{\lstC}{x}{l})}}{\appT{\shrp{\fold}}{f}{y}{l}} \),
  \item\label{itm:gcdlistGcdfGcdfM}%
    \( \dpwc{\appB{\shrp{\gcdf}}{m}{n}}{\appB{\shrp{\gcdf}}{(- m)}{n}}{m < 0} \),
  \item\label{itm:gcdlistGcdfGcdfN}%
    \( \dpwc{\appB{\shrp{\gcdf}}{m}{n}}{\appB{\shrp{\gcdf}}{m}{(- n)}}{n < 0} \), and
  \item\label{itm:gcdlistGcdfGcdfMod}%
    \( \dpwc{\appB{\shrp{\gcdf}}{m}{n}}{\appB{\shrp{\gcdf}}{n}{(m \bmod n)}}{m \ge 0 \wedge n > 0} \).
  \end{enumerate*}
  Note that in \eiLabel{\ref{itm:gcdlistGcdlGcdf}},
  \( m^\prime \) and \( n^\prime \) occur on the right-hand side of \( \arrDePa \)
  but not on the left
  while they are \emph{not} required to be instantiated to ground theory terms
  (\( \gtvA = \emptyset \)).
  This is normal for SDPs \cite{fuh:kop:19,kus:sak:07}.
\end{example}

Termination analysis via SDPs is based on the notion of a chain:

\begin{definition}
  Given a set \( \setDePa \) of SDPs and
  a set \( \setRule \) of rewrite rules,
  a \emph{\( (\setDePa, \setRule) \)-chain} is a (finite or infinite) sequence
  \( (\dpcv{\shrp{\trmB_0}}{\shrp{\trmA_0}}{\rulC_0}{\gtvA_0}, \subA_0), (\dpcv{\shrp{\trmB_1}}{\shrp{\trmA_1}}{\rulC_1}{\gtvA_1}, \subA_1), \ldots \)
  such that for all \( i \),
  \( \dpcv{\shrp{\trmB_i}}{\shrp{\trmA_i}}{\rulC_i}{\gtvA_i} \in \setDePa \),
  \( \subA_i \) is a substitution which respects
  \( \dpcv{\shrp{\trmB_i}}{\shrp{\trmA_i}}{\rulC_i}{\gtvA_i} \), and
  \( \shrp{\trmA_{i - 1}} \subA_{i - 1} \arrRwrt_\setRule^* \shrp{\trmB_i} \subA_i \) if \( i > 0 \).
  The above \( (\setDePa, \setRule) \)-chain is called \emph{computable} if
  \( \trmC \subA_i \) is \( \setReCa \)-computable
  for all \( i \) and \( \trmC \) such that \( \trmA_i \relSupp \trmC \).
\end{definition}

\begin{example}
  Following \cref{ex:gcdlist},
  \( (\text{\ref{itm:gcdlistGcdlGcdf}}, [\sbst{l}{\lstN}, \sbst{m}{42}, \sbst{n}{24}]),
     (\text{\ref{itm:gcdlistGcdfGcdfMod}}, [\sbst{m}{42}, \sbst{n}{24}]),
     (\text{\ref{itm:gcdlistGcdfGcdfMod}}, [\sbst{m}{24}, \sbst{n}{18}]),
     (\text{\ref{itm:gcdlistGcdfGcdfMod}}, [\sbst{m}{18}, \sbst{n}{6}]) \)
  is a computable \( (\setStDP{\setRule}, \setRule) \)-chain.
\end{example}

The key to establishing termination is the following result:

\begin{restatable}{theorem}{chainTermination}\label{thm:chainTermination}
  An AFP system \( \setRule \) is terminating
  if there exists no infinite computable \( (\setStDP{\setRule}, \setRule) \)-chain.
\end{restatable}

The proof (see \cref{app:proofChainTermination})
is very similar to that for unconstrained SDPs \cite{kus:sak:07,fuh:kop:19}.

\section{The Constrained DP Framework}\label{sec:dpframework}
In this section,
we present several techniques based on SDPs,
each as a class of \emph{DP processors};
formally, we call this collection of DP processors
the \emph{constrained (static) DP framework}.
In general,
a DP framework \cite{gie:thi:sch:05,fuh:kop:19}
constitutes a broad method for termination and non-termination.
The presentation here is not complete---%
for example, we do not consider non-termination---%
and a complete one is beyond the scope of this paper.
We rather focus on the most essential DP processors and
those newly designed to handle logical constraints.

For presentation,
we fix an LCSTRS \( \setRule \).

\begin{definition}
  A \emph{DP problem} is a set \( \setDePa \) of SDPs.
  A DP problem \( \setDePa \) is called \emph{finite} if
  there exists no infinite computable \( (\setDePa, \setRule) \)-chain.
  A \emph{DP processor} is a partial mapping
  which possibly assigns to a DP problem a set of DP problems.
  A DP processor \( \dppA \) is called \emph{sound} if
  a DP problem \( \setDePa \) is finite whenever
  \( \dppA(\setDePa) \) consists only of finite DP problems.
\end{definition}

Following \cref{thm:chainTermination},
in order to establish the termination of an AFP system \( \setRule \),
it suffices to show that \( \setStDP{\setRule} \) is a finite DP problem.
Given a collection of sound DP processors,
we have the following procedure:\ %
\begin{enumerate*}
\item \( Q \coloneq \setBrac{\setStDP{\setRule}} \);
\item while \( Q \) contains a DP problem \( \setDePa \)
  to which some sound DP processor \( \dppA \) is applicable,
  \( Q \coloneq (Q \setminus \setBrac{\setDePa}) \cup \dppA(\setDePa) \).
\end{enumerate*}
If this procedure ends with \( Q = \emptyset \),
we can conclude that \( \setRule \) is terminating.

\subsection{The DP Graph and Its Approximations}\label{subsec:dpgraph}
The interconnection of SDPs via chains gives rise to a graph,
namely, the DP graph \cite{art:gie:00},
which models the reachability between dependency pairs.
Since this graph is not computable in general,
we follow the usual convention and consider its (over-)approximations:

\begin{definition}
  Given a set \( \setDePa \) of SDPs,
  a \emph{graph approximation} \( (\grhA_\homA, \homA) \) for \( \setDePa \) consists of
  a finite directed graph \( \grhA_\homA \) and a mapping \( \homA \)
  which assigns to each SDP in \( \setDePa \) a vertex of \( \grhA_\homA \)
  so that there is an edge from \( \homA(\sdpA_0) \) to \( \homA(\sdpA_1) \)
  whenever \( (\sdpA_0, \subA_0), (\sdpA_1, \subA_1) \) is a \( (\setDePa, \setRule) \)-chain
  for some substitutions \( \subA_0 \) and \( \subA_1 \).
\end{definition}
Here \( (\grhA_\homA, \homA) \) approximates the true DP graph
by allowing \( \homA \) to assign a single vertex
to multiple (possibly, infinitely many) SDPs,
and by allowing \( \grhA_\homA \) to contain
an edge from \( \homA(\sdpA_0) \) to \( \homA(\sdpA_1) \)
even if \( \sdpA_0 \) and \( \sdpA_1 \) are not connected
by any \( (\setDePa, \setRule) \)-chain.
In practice, we typically deal with only a finite set \( \setDePa \) of SDPs,
in which case we usually take a bijection for \( \homA \).

This graph structure is useful
because we can leverage it to decompose the DP problem.

\begin{definition}
  Given a DP problem \( \setDePa \),
  a \emph{graph processor} computes a graph approximation \( (\grhA_\homA, \homA) \) for \( \setDePa \) and
  the strongly connected components (SCCs) of \( \grhA_\homA \),
  then returns
  \( \setComp{\setComp{\sdpA \in \setDePa}{\homA(\sdpA) \text{ belongs to } \sccA}}{\sccA \text{ is a non-trivial SCC of } \grhA_\homA} \).
\end{definition}

\begin{example}\label{ex:gcdlistGraph}
  Following \cref{ex:gcdlist},
  a (tight) graph approximation for \( \setStDP{\setRule} \) is in \cref{fig:graph}.
  If a graph processor produces this graph as the graph approximation,
  it will return the set of DP problems
  \( \setBrac{\setBrac{\text{\ref{itm:gcdlistFoldFold}}},
              \setBrac{\text{\ref{itm:gcdlistGcdfGcdfM}}, \text{\ref{itm:gcdlistGcdfGcdfN}}},
              \setBrac{\text{\ref{itm:gcdlistGcdfGcdfMod}}}} \).
\end{example}

\begin{figure}[h]
  \centering
  \begin{tikzpicture}
    \node (GcdlGcdf)    [shape=circle,draw]                    {\ref{itm:gcdlistGcdlGcdf}};
    \node (GcdlFold)    [shape=circle,draw,right=of GcdlGcdf]  {\ref{itm:gcdlistGcdlFold}};
    \node (FoldFold)    [shape=circle,draw,right=of GcdlFold]  {\ref{itm:gcdlistFoldFold}};
    \node (GcdfGcdfN)   [shape=circle,draw,below=of GcdlGcdf]  {\ref{itm:gcdlistGcdfGcdfN}};
    \node (GcdfGcdfM)   [shape=circle,draw,left=of GcdfGcdfN]  {\ref{itm:gcdlistGcdfGcdfM}};
    \node (GcdfGcdfMod) [shape=circle,draw,right=of GcdfGcdfN] {\ref{itm:gcdlistGcdfGcdfMod}};
    \path
    (GcdlGcdf)    edge[->]            (GcdfGcdfM)
    (GcdlGcdf)    edge[->]            (GcdfGcdfN)
    (GcdlGcdf)    edge[->]            (GcdfGcdfMod)
    (GcdlFold)    edge[->]            (FoldFold)
    (FoldFold)    edge[->,loop right] (FoldFold)
    (GcdfGcdfM)   edge[->,bend right] (GcdfGcdfN)
    (GcdfGcdfM)   edge[->,bend right] (GcdfGcdfMod)
    (GcdfGcdfN)   edge[->,bend right] (GcdfGcdfM)
    (GcdfGcdfN)   edge[->]            (GcdfGcdfMod)
    (GcdfGcdfMod) edge[->,loop right] (GcdfGcdfMod);
  \end{tikzpicture}
  \caption{A graph approximation for \( \setStDP{\setRule} \) from \cref{ex:gcdlist}.}
  \label{fig:graph}
\end{figure}

\subparagraph*{Implementation.}
To compute a graph approximation,
we adapt the common \textsc{Cap} approach \cite{gie:thi:sch:fal:06,thi:07} and
take theories into account.
Considering theories allows us, for example,
\emph{not} to have an edge
from \eiLabel{\ref{itm:gcdlistGcdfGcdfMod}} to \eiLabel{\ref{itm:gcdlistGcdfGcdfM}}
in \cref{fig:graph}.

We assume given a finite set of SDPs and
let \( \homA \) be a bijection.
Whether there is an edge
from \( \homA(\dpcv{\shrp{\trmB_0}}{\shrp{\trmA_0}}{\rulC_0}{\gtvA_0}) \)
to \( \homA(\dpcv{\shrp{\trmB_1}}{\shrp{\trmA_1}}{\rulC_1}{\gtvA_1}) \)---%
we rename variables if necessary to avoid name collisions between the two SDPs---%
is determined by the satisfiability (which we check by an SMT solver) of
\( \rulC_0 \wedge \rulC_1 \wedge \tcap(\shrp{\trmA_0}, \shrp{\trmB_1}) \)
where \( \tcap(\trmC, \trmD) \) is defined as follows:
\begin{itemize}
\item If \( \trmC = \appl{\funA}{\vect{\trmC}{1}{n}{\cdots}} \)
  where \( \funA \in \shrp{\setFunc} \) and
  no rewrite rule in \( \setRule \) takes the form
  \( \rwrl{\appl{\funA}{\vect{\rulL}{1}{k}{\cdots}}}{\rulR}{\rulC} \) for \( k \le n \),
  we define \( \tcap(\trmC, \trmD) \) in two cases:
  \begin{bracketenumerate}
  \item \( \tcap(\trmC, \trmD) = \tcap(\trmC_1, \trmD_1) \wedge \cdots \wedge \tcap(\trmC_n, \trmD_n) \)
    if \( \trmD = \appl{\funA}{\vect{\trmD}{1}{n}{\cdots}} \).
  \item \( \tcap(\trmC, \trmD) = \booF \)
    if \( \trmD = \appl{\funB}{\vect{\trmD}{1}{m}{\cdots}} \)
    for some function symbol \( \funB \) other than \( \funA \), and
    either \( \funA \) is not a theory symbol
    or \( \funB \) is not a value.
  \end{bracketenumerate}
\item Suppose \( \tcap(\trmC, \trmD) \) is not defined above;
  \( \tcap(\trmC, \trmD) = (\trmC \relThEq \trmD) \)
  if \( \trmC \in \setTerm{\setTheo{\setFunc}}{\gtvA_0} \) has a base type
  and \( \trmD \) is a theory term in which the type of each variable is a theory sort,
  and \( \tcap(\trmC, \trmD) = \booT \) otherwise.
\end{itemize}
See \cref{app:proofDpframeworkSoundness} for the proof
that this approach produces a graph approximation.

Then strongly connected components can be computed by Tarjan's algorithm \cite{tar:72}.

\begin{example}
  In \cref{fig:graph}, since
  \( (m_0 \ge 0 \wedge n_0 > 0) \wedge m_1 < 0 \wedge (n_0 \relThEq m_1 \wedge m_0 \bmod n_0 \relThEq n_1) \)
  is unsatisfiable,
  there is no edge
  from \eiLabel{\ref{itm:gcdlistGcdfGcdfMod}} to \eiLabel{\ref{itm:gcdlistGcdfGcdfM}}.
\end{example}

\subsection{The Subterm Criterion}
The subterm criterion \cite{hir:mid:04,kus:sak:07}
handles structural recursion and
allows us to remove decreasing SDPs
without considering rewrite rules in \( \setRule \).
We start with defining projections:

\begin{definition}\label{def:projection}
  Let \( \setHead{\setDePa} \) denote the set of function symbols
  heading either side of an SDP in \( \setDePa \).
  A \emph{projection} \( \prjA \) for a set \( \setDePa \) of SDPs
  is a mapping from \( \setHead{\setDePa} \) to integers such that
  \( 1 \le \prjA(\shrp{\funA}) \le n \) if
  \( \shrp{\funA} \relType \vect{\typA}{1}{n}{\arrType \cdots \arrType} \arrType \dpso \).
  Let \( \bar{\prjA}(\appl{\shrp{\funA}}{\vect{\trmA}{1}{n}{\cdots}}) \) denote
  \( \trmA_{\prjA(\shrp{\funA})} \).
\end{definition}

A projection chooses an argument position for each relevant function symbol
so that arguments at those positions do not increase in a chain.

\begin{definition}
  Given a set \( \setDePa \) of SDPs,
  a projection \( \prjA \) is said to
  \emph{\( \relSupp \)-orient} a subset \( \setDePa^\prime \) of \( \setDePa \) if
  \( \bar{\prjA}(\shrp{\trmB}) \relSupp \bar{\prjA}(\shrp{\trmA}) \) for all
  \( \dpcv{\shrp{\trmB}}{\shrp{\trmA}}{\rulC}{\gtvA} \in \setDePa^\prime \) and
  \( \bar{\prjA}(\shrp{\trmB}) = \bar{\prjA}(\shrp{\trmA}) \) for all
  \( \dpcv{\shrp{\trmB}}{\shrp{\trmA}}{\rulC}{\gtvA} \in \setDePa \setminus \setDePa^\prime \).
  A \emph{subterm criterion processor}
  assigns to a DP problem \( \setDePa \)
  the singleton \( \setBrac{\setDePa \setminus \setDePa^\prime} \)
  for some non-empty subset \( \setDePa^\prime \) of \( \setDePa \) such that
  there exists a projection for \( \setDePa \) which \( \relSupp \)-orients \( \setDePa^\prime \).
\end{definition}

\begin{example}
  Following \cref{ex:gcdlistGraph},
  a subterm criterion processor is applicable to \( \setBrac{\text{\ref{itm:gcdlistFoldFold}}} \).
  Let \( \prjA(\shrp{\fold}) \) be \( 3 \) so that
  \( \bar{\prjA}(\appT{\shrp{\fold}}{f}{y}{(\appB{\lstC}{x}{l})})
  = \appB{\lstC}{x}{l} \relSupp l
  = \bar{\prjA}(\appT{\shrp{\fold}}{f}{y}{l}) \).
  The processor returns \( \setBrac{\emptyset} \),
  and the empty DP problem can (trivially) be removed by a graph processor.
\end{example}

\subparagraph*{Implementation.}
The search for a suitable projection can be done through SMT and is standard:\ %
we introduce an integer variable \( N_{\shrp{\funA}} \)
that represents \( \prjA(\shrp{\funA}) \)
for each \( \shrp{\funA} \in \setHead{\setDePa} \),
and a boolean variable \( \mathsf{strict}_\sdpA \)
for each \( \sdpA \in \setDePa \);
then we encode the requirement per SDP.

\subsection{Integer Mappings}
The subterm criterion deals with recursion over the structure of terms,
but not recursion over, say, integers,
which requires us to utilize the information in logical constraints.
In this subsection,
we assume that \( \tyIn \in \setTheo{\setSort} \) and
\( \setTheo{\setFunc} \supseteq \setBrac{{\ge}, {>}, {\wedge}} \),
where \( {\ge} \relType \tyIn \arrType \tyIn \arrType \tyBo \),
\( {>} \relType \tyIn \arrType \tyIn \arrType \tyBo \) and
\( {\wedge} \relType \tyBo \arrType \tyBo \arrType \tyBo \)
are interpreted in the standard way.

\begin{definition}
  Given a set \( \setDePa \) of SDPs,
  for all \( \shrp{\funA} \in \setHead{\setDePa} \) (see \cref{def:projection}) where
  \( \shrp{\funA} \relType \vect{\typA}{1}{n}{\arrType \cdots \arrType} \arrType \dpso \),
  let \( \setFreI{\shrp{\funA}} \) be the subset of \( \setBrac{1, \ldots, n} \) such that
  \( i \in \setFreI{\shrp{\funA}} \) if and only if
  \( \typA_i \in \setTheo{\setSort} \) and
  the \( i \)-th argument of any occurrence of \( \shrp{\funA} \) in an SDP
  \( \dpcv{\shrp{\trmB}}{\shrp{\trmA}}{\rulC}{\gtvA} \in \setDePa \)
  is in \( \setTerm{\setTheo{\setFunc}}{\gtvA} \).
  Let \( \setFreV{\shrp{\funA}} \) be a set of fresh variables
  \( \setComp{\varA_{\shrp{\funA}, i}}{i \in \setFreI{\shrp{\funA}}} \)
  where \( \varA_{\shrp{\funA}, i} \relType \typA_i \) for all \( i \).
  An \emph{integer mapping} \( \inmA \) for \( \setDePa \) is a mapping
  from \( \setHead{\setDePa} \) to theory terms such that
  for all \( \shrp{\funA} \),
  \( \inmA(\shrp{\funA}) \relType \tyIn \) and
  \( \setFvar{\inmA(\shrp{\funA})} \subseteq \setFreV{\shrp{\funA}} \).
  Let \( \bar{\inmA}(\appl{\shrp{\funA}}{\vect{\trmA}{1}{n}{\cdots}}) \) denote
  \( \inmA(\shrp{\funA})[\sbst{\varA_{\shrp{\funA}, i}}{\trmA_i}]_{i \in \setFreI{\shrp{\funA}}} \).
\end{definition}

With integer mappings, we can handle decreasing integer values.

\begin{definition}
  Given a set \( \setDePa \) of SDPs,
  an integer mapping \( \inmA \) is said to
  \emph{\( > \)-orient} a subset \( \setDePa^\prime \) of \( \setDePa \)
  if \( \rulC \models \bar{\inmA}(\shrp{\trmB}) \ge 0 \wedge \bar{\inmA}(\shrp{\trmB}) > \bar{\inmA}(\shrp{\trmA}) \)
  for all \( \dpcv{\shrp{\trmB}}{\shrp{\trmA}}{\rulC}{\gtvA} \in \setDePa^\prime \),
  and \( \rulC \models \bar{\inmA}(\shrp{\trmB}) \ge \bar{\inmA}(\shrp{\trmA}) \)
  for all \( \dpcv{\shrp{\trmB}}{\shrp{\trmA}}{\rulC}{\gtvA} \in \setDePa \setminus \setDePa^\prime \),
  where \( \rulC \models \rulC^\prime \) denotes that
  \( \mapNtrp{\rulC \subA} = \algT \) implies \( \mapNtrp{\rulC^\prime \subA} = \algT \)
  for each substitution \( \subA \)
  which maps variables in \( \setFvar{\rulC} \cup \setFvar{\rulC^\prime} \) to values.
  An \emph{integer mapping processor} assigns to a DP problem \( \setDePa \)
  the singleton \( \setBrac{\setDePa \setminus \setDePa^\prime} \)
  for some non-empty subset \( \setDePa^\prime \) of \( \setDePa \) such that
  there exists an integer mapping for \( \setDePa \) which \( > \)-orients \( \setDePa^\prime \).
\end{definition}

\begin{example}
  Following \cref{ex:gcdlistGraph}, an integer mapping processor is applicable to
  \( \setBrac{\text{\ref{itm:gcdlistGcdfGcdfMod}}} \).
  Let \( \inmA(\shrp{\gcdf}) \) be \( \varA_{\shrp{\gcdf}, 2} \) so that
  \( \bar{\inmA}(\appB{\shrp{\gcdf}}{m}{n}) = n \),
  \( \bar{\inmA}(\appB{\shrp{\gcdf}}{n}{(m \bmod n)}) = m \bmod n \) and
  \( m \ge 0 \wedge n > 0 \models n \ge 0 \wedge n > m \bmod n \).
  The processor returns \( \setBrac{\emptyset} \),
  and the empty DP problem can (trivially) be removed by a graph processor.
\end{example}

\subparagraph*{Implementation.}
There are several ways to implement integer mapping processors.
In our implementation,
we generate a number of ``interpretation candidates'' and
use an SMT encoding to select for each \( \shrp{\funA} \in \setHead{\setDePa} \)
one candidate that satisfies the requirements.
Candidates include forms such as \( \inmA(\shrp{\funA}) = \varA_{\shrp{\funA}, i} \) and
those that are generated from the SDPs' logical constraints---e.g.,
given \( \dpwc{\appB{\shrp{\mathsf{f}}}{x}{y}}{\appB{\shrp{\mathsf{g}}}{x}{(y+1)}}{y < x} \),
we generate \( \inmA(\shrp{\mathsf{f}}) = \varA_{\shrp{\mathsf{f}}, 1} - \varA_{\shrp{\mathsf{f}}, 2} - 1 \)
because \( y < x \) implies \( x - y - 1 \ge 0 \).

\subsection{Theory Arguments}\label{subsec:theoryarg}
Integer mapping processors have a clear limitation:\ %
what if some key variables do not occur in the set \( \gtvA \)?
This is observed in the remaining DP problem
\( \setBrac{\text{\ref{itm:gcdlistGcdfGcdfM}}, \text{\ref{itm:gcdlistGcdfGcdfN}}} \)
from \cref{ex:gcdlist}.
It is clearly finite but
no integer mapping processor is applicable
since \( \setFreI{\shrp{\gcdf}} = \emptyset \).

This restriction exists for a reason.
Variables that are not guaranteed to be instantiated to theory terms
may be instantiated to \emph{non-deterministic} terms---e.g.,
\( \setBrac{\dpwc{\appT{\shrp{\mathsf{f}}}{x}{y}{z}}{\appT{\shrp{\mathsf{f}}}{x}{(x + 1)}{(x - 1)}}{y < z}} \)
is not a finite DP problem if
\( \setRule \supseteq \setBrac{\appB{\mathsf{c}}{x}{y} \arrRwrt x, \appB{\mathsf{c}}{x}{y} \arrRwrt y} \).

The problem of \( \setBrac{\text{\ref{itm:gcdlistGcdfGcdfM}}, \text{\ref{itm:gcdlistGcdfGcdfN}}} \)
arises because each SDP focuses on only one argument:\ %
for example, the logical constraint (with the component \( \gtvA \))
of \eiLabel{\ref{itm:gcdlistGcdfGcdfN}} only concerns \( n \) so
in principle we cannot assume anything about \( m \).
Yet, if \eiLabel{\ref{itm:gcdlistGcdfGcdfN}}
follows \eiLabel{\ref{itm:gcdlistGcdfGcdfM}} in a chain,
we \emph{can} derive that \( m \) must be instantiated to a ground theory term
(we call such an argument a \emph{theory argument}).
We explore a way of propagating this information.

\begin{definition}
  A \emph{theory argument (position) mapping} \( \tamA \) for a set \( \setDePa \) of SDPs
  is a mapping from \( \setHead{\setDePa} \) (see \cref{def:projection})
  to subsets of \( \setInte \) such that
  \( \tamA(\shrp{\funA}) \subseteq \setComp{1 \le i \le m}{\typA_i \in \setTheo{\setSort}} \)
  if \( \shrp{\funA} \relType \vect{\typA}{1}{m}{\arrType \cdots \arrType} \arrType \dpso \),
  \( \trmB_i \) is a theory term and
  the type of each variable in \( \setFvar{\trmB_i} \) is a theory sort for all
  \( \dpcv{\appl{\shrp{\funA}}{\vect{\trmB}{1}{m}{\cdots}}}{\shrp{\trmA}}{\rulC}{\gtvA} \in \setDePa \) and
  \( i \in \tamA(\shrp{\funA}) \),
  and \( \trmA_j \) is a theory term and
  \( \setFvar{\trmA_j} \subseteq \gtvA \cup \bigcup_{i \in \tamA(\shrp{\funA})} \setFvar{\trmB_i} \) for all
  \( \dpcv{\appl{\shrp{\funA}}{\vect{\trmB}{1}{m}{\cdots}}}{\appl{\shrp{\funB}}{\vect{\trmA}{1}{n}{\cdots}}}{\rulC}{\gtvA} \in \setDePa \) and
  \( j \in \tamA(\shrp{\funB}) \).
  Let \( \bar{\tamA}(\dpcv{\appl{\shrp{\funA}}{\vect{\trmB}{1}{m}{\cdots}}}{\shrp{\trmA}}{\rulC}{\gtvA}) \) denote
  \( \dpcv{\appl{\shrp{\funA}}{\vect{\trmB}{1}{m}{\cdots}}}{\shrp{\trmA}}{\rulC}{\gtvA \cup \bigcup_{i \in \tamA(\shrp{\funA})} \setFvar{\trmB_i}} \).
\end{definition}

By a theory argument mapping,
we choose a subset of the given set of SDPs
from which the theory argument information is propagated.

\begin{definition}
  Given a set \( \setDePa \) of SDPs,
  a theory argument mapping \( \tamA \) is said to
  \emph{fix} a subset \( \setDePa^\prime \) of \( \setDePa \) if
  \( \bigcup_{i \in \tamA(\shrp{\funA})} \setFvar{\trmA_i} \subseteq \gtvA \) for all
  \( \dpcv{\shrp{\trmB}}{\appl{\shrp{\funA}}{\vect{\trmA}{1}{n}{\cdots}}}{\rulC}{\gtvA} \in \setDePa^\prime \).
  A \emph{theory argument processor}
  assigns to a DP problem \( \setDePa \)
  the pair \( \setBrac{\setComp{\bar{\tamA}(\sdpA)}{\sdpA \in \setDePa}, \setDePa \setminus \setDePa^\prime} \)
  for some non-empty subset \( \setDePa^\prime \) of \( \setDePa \) such that
  there exists a theory argument mapping for \( \setDePa \) which fixes \( \setDePa^\prime \).
\end{definition}

\begin{example}
  Following \cref{ex:gcdlistGraph},
  a theory argument processor is applicable to
  \( \setBrac{\text{\ref{itm:gcdlistGcdfGcdfM}}, \text{\ref{itm:gcdlistGcdfGcdfN}}} \).
  Let \( \tamA(\shrp{\gcdf}) \) be \( \setBrac{1} \) so that
  \( \tamA \) fixes \( \setBrac{\text{\ref{itm:gcdlistGcdfGcdfM}}} \).
  The processor returns the pair
  \( \setBrac{\setBrac{\text{\ref{itm:gcdlistGcdfGcdfM}},
                       \text{\eiLabel{7} } \dpcv{\appB{\shrp{\gcdf}}{m}{n}}{\appB{\shrp{\gcdf}}{m}{(- n)}}{n < 0}{\setBrac{m, n}}},
              \setBrac{\text{\ref{itm:gcdlistGcdfGcdfN}}}} \).
  The integer mapping processor with \( \inmA(\shrp{\gcdf}) = {- \varA_{\shrp{\gcdf}, 1}} \)
  removes \eiLabel{\ref{itm:gcdlistGcdfGcdfM}} from
  \( \setBrac{\text{\ref{itm:gcdlistGcdfGcdfM}}, \text{7}} \).
  Then \( \setBrac{\text{7}} \) and
  \( \setBrac{\text{\ref{itm:gcdlistGcdfGcdfN}}} \)
  can be removed by graph processors.
\end{example}

\subparagraph*{Implementation.}
To find a valid theory argument mapping,
we simply start by setting \( \tamA(\shrp{\funA}) = \setBrac{1, \ldots, m} \)
for all \( \shrp{\funA} \),
and choose one SDP to fix.
Then we iteratively remove arguments that do not satisfy the condition
until no such argument is left.

\subsection{Reduction Pairs}
Although it is not needed by the running example,
we present a constrained variant of \emph{reduction pair processors},
which are at the heart of most unconstrained termination analyzers.

\begin{definition}
  A \emph{constrained relation} \( \relA \) is a set of quadruples
  \( (\trmB, \trmA, \rulC, \gtvA) \) where
  \( \trmB \) and \( \trmA \) are terms which have the same type,
  \( \rulC \) is a logical constraint and
  \( \gtvA \supseteq \setFvar{\rulC} \) is a set of variables
  whose types are theory sorts.
  We write \( \crel{\relA}{\trmB}{\trmA}{\rulC}{\gtvA} \)
  if \( (\trmB, \trmA, \rulC, \gtvA) \in \relA \).
  A binary relation \( \relA^\prime \) over terms is said to
  \emph{cover} a constrained relation \( \relA \) if
  \( \crel{\relA}{\trmB}{\trmA}{\rulC}{\gtvA} \) implies that
  \( \mapNorm{(\trmB \subA)}_\lblCalc \mathrel{\relA^\prime} \mapNorm{(\trmA \subA)}_\lblCalc \)
  for each substitution \( \subA \) such that
  \( \subA(\varA) \) is a ground theory term for all \( \varA \in \gtvA \) and
  \( \mapNtrp{\rulC \subA} = \algT \).

  A \emph{constrained reduction pair} \( (\succeq, \succ) \)
  is a pair of constrained relations such that
  \( \succeq \) is covered by some reflexive relation \( \sqsupseteq \)
  which includes \( \arrRwrt_\lblCalc \) and is monotonic
  (i.e., \( \trmB \sqsupseteq \trmA \) implies \( \plug{\ctxA}{\trmB} \sqsupseteq \plug{\ctxA}{\trmA} \)),
  \( \succ \) is covered by some well-founded relation \( \sqsupset \),
  and \( \relComp{\sqsupseteq}{\sqsupset} \subseteq {\sqsupset^+} \).
\end{definition}

\begin{definition}
  A \emph{reduction pair processor} assigns to a DP problem \( \setDePa \)
  the singleton \( \setBrac{\setDePa \setminus \setDePa^\prime} \)
  for some non-empty subset \( \setDePa^\prime \) of \( \setDePa \) such that
  there exists a constrained reduction pair \( (\succeq, \succ) \) where
  \begin{enumerate*}
  \item \( \crel{\succ}{\shrp{\trmB}}{\shrp{\trmA}}{\rulC}{\gtvA} \)
    for all \( \dpcv{\shrp{\trmB}}{\shrp{\trmA}}{\rulC}{\gtvA} \in \setDePa^\prime \),
  \item \( \crel{\succeq}{\shrp{\trmB}}{\shrp{\trmA}}{\rulC}{\gtvA} \)
    for all \( \dpcv{\shrp{\trmB}}{\shrp{\trmA}}{\rulC}{\gtvA} \in \setDePa \setminus \setDePa^\prime \), and
  \item \( \crel{\succeq}{\rulL}{\rulR}{\rulC}{\setFvar{\rulC} \cup (\setFvar{\rulR} \setminus \setFvar{\rulL})} \)
    for all \( \rwrl{\rulL}{\rulR}{\rulC} \in \setRule \).
  \end{enumerate*}
\end{definition}

While a variety of reduction pairs have been proposed for unconstrained rewriting,
it is not yet the case in a higher-order and constrained setting:\ %
so far the only one is a limited version of HORPO \cite{guo:kop:24},
which is adapted into a weakly monotonic reduction pair \cite{kop:24}
and then implemented in the DP framework.
This is still a prototype definition.

We have included reduction pair processors here
because their definition allows us to start designing constrained reduction pairs.
In particular, as unconstrained reduction pairs
can be used as the covering pair \( (\sqsupseteq, \sqsupset) \),
it is likely that many of them
(such as variants of HORPO and weakly monotonic algebras)
can be adapted.
\bigbreak

We conclude this section by the following result (see \cref{app:proofDpframeworkSoundness}):

\begin{theorem}\label{thm:dpframeworkSoundness}
  All the DP processors defined in \cref{sec:dpframework} are sound.
\end{theorem}

\section{Universal Computability}\label{sec:computability}
Termination is not a \emph{modular} property:\ %
given terminating systems \( \setRule_0 \) and \( \setRule_1 \),
we cannot generally conclude that \( \setRule_0 \cup \setRule_1 \) is also terminating.
As computability is based on termination,
it is not modular either.
For example, both
\( \setBrac{\mathsf{a} \arrRule \mathsf{b}} \) and
\( \setBrac{\appl{\mathsf{f}}{\mathsf{b}} \arrRule \appl{\mathsf{f}}{\mathsf{a}}} \)
are terminating,
and \( \mathsf{f} \relType \mathsf{o} \arrType \mathsf{o} \)
is computable in the second system;
yet, combining the two yields
\( \appl{\mathsf{f}}{\mathsf{a}} \arrRwrt
   \appl{\mathsf{f}}{\mathsf{b}} \arrRwrt
   \appl{\mathsf{f}}{\mathsf{a}} \arrRwrt
   \cdots \),
which refutes the termination of the combination and
the computability of \( \mathsf{f} \).

On the other hand,
functions like \( \mapf \) and \( \fold \) are prevalently used;
the lack of a modular principle for
the termination analysis of higher-order systems involving such functions is painful.
Moreover, if such a system is non-terminating,
this is seldom attributed to those functions,
which are generally considered ``terminating''
regardless of how they may be called.

In this section,
we propose \emph{universal computability},
a concept which corresponds to the termination of a function
in all ``reasonable'' uses.
First, we rephrase the notion of a hierarchical combination
\cite{rao:93,rao:94,rao:95,der:95} in terms of LCSTRSs:

\begin{definition}
  An LCSTRS \( \setRule_1 \) is called an \emph{extension} of
  a base system \( \setRule_0 \) if
  the two systems' interpretations of theory symbols
  coincide over all the theory symbols in common,
  and function symbols in \( \setRule_0 \) are not defined by
  any rewrite rule in \( \setRule_1 \).
  Given a base system \( \setRule_0 \) and
  an extension \( \setRule_1 \) of \( \setRule_0 \),
  the system \( \setRule_0 \cup \setRule_1 \) is called a \emph{hierarchical combination}.
\end{definition}
In a hierarchical combination,
function symbols in the base system can occur in the extension,
but cannot be (re)defined.
This forms the basis of the modular programming scenario we are interested in:\ %
think of the base system as a library containing the definitions of,
say, \( \mapf \) and \( \fold \).
We further define a class of extensions to take information hiding into account:

\begin{definition}
  Given an LCSTRS \( \setRule_0 \) and
  a set of function symbols---called
  \emph{hidden symbols}---in \( \setRule_0 \),
  an extension \( \setRule_1 \) of \( \setRule_0 \)
  is called a \emph{public extension} if
  hidden symbols do not occur in
  any rewrite rule in \( \setRule_1 \).
\end{definition}

Now we present the central definitions of this section:

\begin{definition}
  Given an LCSTRS \( \setRule_0 \) with a sort ordering \( \relSord \),
  a term \( \trmA \) is called \emph{universally computable} if
  for each extension \( \setRule_1 \) of \( \setRule_0 \) and
  each extension \( \relSord^\prime \) of \( \relSord \)
  to sorts in \( \setRule_0 \cup \setRule_1 \)
  (i.e., \( \relSord^\prime \) coincides with \( \relSord \) over sorts in \( \setRule_0 \)),
  \( \trmA \) is \( \setReCa \)-computable in
  \( \setRule_0 \cup \setRule_1 \) with \( \relSord^\prime \);
  if a set of hidden symbols in \( \setRule_0 \) is also given and
  the above universal quantification of \( \setRule_1 \)
  is restricted to \emph{public} extensions,
  such a term \( \trmA \) is called \emph{publicly computable}.

  The base system \( \setRule_0 \) is called universally computable if
  all its terms are;
  it is called publicly computable if
  all its \emph{public} terms---terms that contain no hidden symbol---are.
\end{definition}
With an empty set of hidden symbols,
the two notions---universal computability and public computability---coincide.
Below we state common properties in terms of public computability.

In summary, we consider passing \( \setReCa \)-computable arguments to
a defined symbol in \( \setRule_0 \)
the ``reasonable'' way of calling the function.
To establish the universal computability of higher-order functions
such as \( \mapf \) and \( \fold \)---%
i.e., to prove that they are \( \setReCa \)-computable in
\emph{all} relevant hierarchical combinations---%
we will use SDPs, which are about \( \setReCa \)-computability.

\begin{example}
  The system \( \setBrac{\appl{\lapp}{(\appl{\llam}{f})} \arrRule f} \)
  in \cref{sec:introduction} is not universally computable
  due to the extension \( \setBrac{\appl{\mathsf{w}}{x} \arrRule \appB{\lapp}{x}{x}} \).
\end{example}

\subsection{The DP Framework Revisited}
To use SDPs for universal---or public---computability,
we need a more general version of \cref{thm:chainTermination}.
We start with defining public chains:

\begin{definition}
  An SDP
  \( \dpcv{\appl{\shrp{\funA}}{\vect{\trmB}{1}{m}{\cdots}}}{\shrp{\trmA}}{\rulC}{\gtvA} \)
  is called public if \( \funA \) is not a hidden symbol.
  A \( (\setDePa, \setRule) \)-chain is called public if
  its first SDP is public.
\end{definition}

Now we state the main result of this section:

\begin{restatable}{theorem}{chainComputability}\label{thm:chainComputability}
  An AFP system \( \setRule_0 \) with sort ordering \( \relSord \)
  is publicly computable with respect to a set of hidden symbols in \( \setRule_0 \)
  if there exists no infinite computable
  \( (\setStDP{\setRule_0}, \setRule_0 \cup \setRule_1) \)-chain that is \emph{public}
  for each public extension \( \setRule_1 \) of \( \setRule_0 \) and
  each extension \( \relSord^\prime \) of \( \relSord \)
  to sorts in \( \setRule_0 \cup \setRule_1 \).
\end{restatable}

While this result is not surprising and
its proof (see \cref{app:proofChainComputability}) is standard,
it is not obvious how it can be used.
The key observation which enables us to use the DP framework for public computability
is that among the DP processors in \cref{sec:dpframework},
only reduction pair processors rely on
the rewrite rules of the underlying system \( \setRule \)
(depending on how it computes a graph approximation,
a graph processor does not have to know all the rewrite rules).
Henceforth, we fix a base system \( \setRule_0 \),
a set of hidden symbols in \( \setRule_0 \) and
an arbitrary, unknown public extension \( \setRule_1 \) of \( \setRule_0 \).
Now \( \setRule \) is the hierarchical combination \( \setRule_0 \cup \setRule_1 \).

First, we generalize the definition of a DP problem:

\begin{definition}
  A \emph{(universal) DP problem} \( (\setDePa, \pflg) \) consists of
  a set \( \setDePa \) of SDPs and
  a flag \( \pflg \in \setBrac{\flagAny, \flagPub} \)
  (for \emph{any} or \emph{public}).
  A DP problem \( (\setDePa, \pflg) \) is called \emph{finite} if either
  \begin{enumerate*}
  \item \( \pflg = \flagAny \) and
    there exists no infinite computable \( (\setDePa, \setRule_0 \cup \setRule_1) \)-chain, or
  \item \( \pflg = \flagPub \) and
    there exists no infinite computable \( (\setDePa, \setRule_0 \cup \setRule_1) \)-chain
    which is \emph{public}.
  \end{enumerate*}
\end{definition}
DP processors are defined in the same way as before,
now for universal DP problems.
The goal is to show that \( (\setStDP{\setRule_0}, \flagPub) \) is finite,
and the procedure for termination in \cref{sec:dpframework} also works here
if we change the initialization of \( Q \) accordingly.

Next, we review the DP processors presented in \cref{sec:dpframework}.
For each \( \dppA \) of the original graph, subterm criterion and integer mapping processors,
the processor \( \dppA^\prime \) such that
\( \dppA^\prime(\setDePa, \pflg) = \setComp{(\setDePa^\prime, \flagAny)}{\setDePa^\prime \in \dppA(\setDePa)} \)
is sound for universal DP problems.
For theory argument processors,
we can do better when the input flag is \( \flagPub \)
(when it is \( \flagAny \),
we just handle \( \setDePa \) in the same way as we do in \cref{sec:dpframework} and
the output flags are obviously \( \flagAny \)):\ %
if the subset \( \setDePa^\prime \) of \( \setDePa \)
fixed by a theory argument mapping \( \tamA \)
contains all the public SDPs in \( \setDePa \),
the processor should return the singleton
\( \setBrac{(\setComp{\sdpA}{\sdpA \in \setDePa\text{ is public}} \cup \setComp{\bar{\tamA}(\sdpA)}{\sdpA \in \setDePa\text{ is not public}}, \flagPub)} \);
otherwise, the pair \( \setBrac{(\setComp{\bar{\tamA}(\sdpA)}{\sdpA \in \setDePa}, \flagAny), (\setDePa \setminus \setDePa^\prime, \flagPub)} \).
Reduction pair processors require knowledge of the extension \( \setRule_1 \)
so we do not adapt them.

\subparagraph*{New Processors.}
Last, we propose two classes of DP processors that are useful for public computability.
Processors of the first class do not actually simplify DP problems;
they rather alter their input to allow other DP processors to be applied subsequently.

\begin{definition}
  Given sets \( \setDePa_1 \) and \( \setDePa_2 \) of SDPs,
  \( \setDePa_2 \) is said to \emph{cover} \( \setDePa_1 \) if
  for each SDP \( \dpcv{\shrp{\trmB}}{\shrp{\trmA}}{\rulC_1}{\gtvA_1} \in \setDePa_1 \) and
  each substitution \( \subA_1 \) which respects \( \dpcv{\shrp{\trmB}}{\shrp{\trmA}}{\rulC_1}{\gtvA_1} \),
  there exist an SDP \( \dpcv{\shrp{\trmB}}{\shrp{\trmA}}{\rulC_2}{\gtvA_2} \in \setDePa_2 \) and
  a substitution \( \subA_2 \) such that
  \( \subA_2 \) respects \( \dpcv{\shrp{\trmB}}{\shrp{\trmA}}{\rulC_2}{\gtvA_2} \),
  \( \trmB \subA_1 = \trmB \subA_2 \) and
  \( \trmA \subA_1 = \trmA \subA_2 \).
  A \emph{constraint modification processor}
  assigns to a DP problem \( (\setDePa, \pflg) \)
  the singleton \( \setBrac{(\setDePa^\prime, \pflg)} \)
  for some \( \setDePa^\prime \) which covers \( \setDePa \).
\end{definition}

Now combined with the information of hidden symbols,
the DP graph allows us to remove SDPs that are unreachable from any public SDP.

\begin{definition}
  A \emph{reachability processor}
  assigns to a DP problem \( (\setDePa, \flagPub) \)
  the singleton
  \( \setBrac{(\setComp{\sdpA \in \setDePa}{\homA(\sdpA) \text{ is reachable from } \homA(\sdpA_0) \text{ for some public SDP } \sdpA_0}, \flagPub)} \),
  given a graph approximation \( (\grhA_\homA, \homA) \) for \( \setDePa \).
\end{definition}

These two classes of DP processors are often used together:\ %
a constraint modification processor can split an SDP into simpler ones,
some of which may be removed by a reachability processor.
In our implementation, a constraint modification processor is particularly used
to break an SDP \( \dpcv{\shrp{\trmB}}{\shrp{\trmA}}{\trmC \neq \trmD}{\gtvA} \)
into two SDPs with logical constraints
\( \trmC < \trmD \) and \( \trmC > \trmD \), respectively (see \cref{ex:constraint});
similarly for \( \dpcv{\shrp{\trmB}}{\shrp{\trmA}}{\trmC \vee \trmD}{\gtvA} \).

\begin{example}\label{ex:constraint}
  Consider an alternative implementation of the factorial function from \cref{ex:fact},
  which has SDPs
  \begin{enumerate*}
  \item\label{itm:constraintFactComp}%
    \( \dpwc{\appB{\shrp{\fact}}{n}{k}}{\appT{\shrp{\comp}}{k}{(\appl{(*)}{n})}{x^\prime}}{n \neq 0} \),
  \item\label{itm:constraintFactFact}%
    \( \dpwc{\appB{\shrp{\fact}}{n}{k}}{\appB{\shrp{\fact}}{(n - 1)}{(\appB{\comp}{k}{(\appl{(*)}{n})})}}{n \neq 0} \), and
  \item\label{itm:constraintInitFact}%
    \( \dpnc{\appl{\shrp{\init}}{k}}{\appB{\shrp{\fact}}{42}{k}} \).
  \end{enumerate*}
  Assume that \( \fact \) is a hidden symbol.
  Note that
  \( (\setBrac{\text{\ref{itm:constraintFactComp}}, \text{\ref{itm:constraintFactFact}}, \text{\ref{itm:constraintInitFact}}}, \flagPub) \)
  is not finite without this assumption.
  A constraint modification processor
  can replace \eiLabel{\ref{itm:constraintFactFact}} with
  \begin{enumerate*}\setcounter{enumi}{3}%
  \item\label{itm:constraintFactFactLt}%
    \( \dpwc{\appB{\shrp{\fact}}{n}{k}}{\appB{\shrp{\fact}}{(n - 1)}{(\appB{\comp}{k}{(\appl{(*)}{n})})}}{n < 0} \), and
  \item\label{itm:constraintFactFactGt}%
    \( \dpwc{\appB{\shrp{\fact}}{n}{k}}{\appB{\shrp{\fact}}{(n - 1)}{(\appB{\comp}{k}{(\appl{(*)}{n})})}}{n > 0} \).
  \end{enumerate*}
  A reachability processor can then remove \eiLabel{\ref{itm:constraintFactFactLt}}.
  The remaining DP problem
  \( (\setBrac{\text{\ref{itm:constraintFactComp}}, \text{\ref{itm:constraintInitFact}}, \text{\ref{itm:constraintFactFactGt}}}, \flagPub) \)
  can easily be handled by a graph processor and an integer mapping processor.
\end{example}
\bigbreak

We conclude this section by the following result (see \cref{app:proofComputabilitySoundness}):

\begin{theorem}\label{thm:computabilitySoundness}
  All the DP processors defined in \cref{sec:computability} are sound.
\end{theorem}

\section{Experiments and Future Work}\label{sec:experiments}
All the results in this paper have been implemented
in our open-source analyzer \cora{} \cite{cora}.
We have evaluated \cora{} on three groups of experiments,
and the results are in \cref{tab:experiments}.

\begin{table}[h]
  \centering
  \caption{\cora{} experiment results.}
  \label{tab:experiments}
  \begin{tabular}{|c|c|c|c|}
    \hline
    & Custom & STRS & ITRS \\
    \hline
    Termination & 20/28 & 72/140 & 69/117 \\
    Computability & 20/28 & 66/140 & 68/117 \\
    Wanda & -- & 105/140 & -- \\
    AProVE & -- & -- & 102/117 \\
    \hline
  \end{tabular}
\end{table}

The first group contains examples in this paper and
several other LC(S)TRS benchmarks we have collected.
The second group contains all the \( \lambda \)-free problems
from the higher-order category of TPDB \cite{tpdb}.
The third group contains problems
from the first-order ``integer TRS innermost'' category.
The computability tests analyze public computability;
since there are no hidden symbols in TPDB,
the main difference from a termination check is
that reduction pair processors are disabled.
A full evaluation page is available through the link:
\begin{center}
  \url{https://www.cs.ru.nl/~cynthiakop/experiments/mfcs2024}
\end{center}

Unsurprisingly, \cora{} is substantially weaker than
Wanda \cite{kop:20} on unconstrained higher-order TRSs,
and AProVE \cite{gie:asc:bro:emm:fro:fuh:hen:ott:plu:sch:str:swi:thi:17}
on first-order integer TRSs:\ %
this work aims to be a starting point for \emph{combining}
higher-order term analysis and theory reasoning,
and cannot yet compete with dedicated tools that have had years of development.
Nevertheless, we believe that these results show a solid foundation
with only a handful of simple techniques.

\subparagraph*{Future Work.}
Many of the existing techniques used in the analyses of integer TRSs and higher-order TRSs
are likely to be extensible to our setting,
leaving many encouraging avenues for further development.
We highlight the most important few:
\begin{itemize}
\item Usable rules with respect to an argument filtering \cite{gie:thi:sch:fal:06,kop:22}.
  To effectively use reduction pairs, being able to discard some rewrite rules is essential
  (especially for universal computability, if we can discard the unknown ones).
  Closely related is the adaptation of more reduction pairs such as
  weakly monotonic algebras \cite{zan:94,pol:96},
  tuple interpretations \cite{kop:val:21,yam:22} and
  more sophisticated path orderings \cite{bla:jou:rub:08},
  all of which have higher-order formulations.
\item Transformation techniques, such as narrowing,
  and chaining dependency pairs together
  (as used for instance for integer transition systems
  \cite[Secion~3.1]{gie:asc:bro:emm:fro:fuh:hen:ott:plu:sch:str:swi:thi:17}).
  This could also be a step toward using the constrained DP framework for non-termination.
\item Handling the innermost or call-by-value strategy.
  Several functional languages adopt call-by-value evaluation,
  and applying this restriction may allow for more powerful analyses.
  In the first-order DP framework, there is ample work on the innermost strategy
  to build on (see, e.g., \cite{gie:thi:sch:05,gie:thi:sch:fal:06}).
\item Theory-specific processors for popular theories other than integers,
  e.g., bit vectors \cite{mat:nis:koj:shi:23}.
\end{itemize}

\bibliography{mfcs2024}

\appendix

\section{Proofs of the Results}

\subsection{The Proof of Theorem~\ref{thm:chainTermination}}\label{app:proofChainTermination}
\Cref{thm:chainTermination} is closely connected to \cref{thm:chainComputability},
which is a more general version of the former,
and the following proof is formulated on the basis of a common premise---%
\cref{cor:chainExistence} (see \cref{app:proofChainComputability})---%
to stress the connection.

\chainTermination*

\begin{proof}
  Assume that \( \setRule \) is not terminating.
  By definition, there exists a non-terminating term \( \trmC \),
  which cannot be \( \setReCa \)-computable.
  Take a minimal subterm \( \trmB \) of \( \trmC \)
  such that \( \trmB \) is not \( \setReCa \)-computable,
  then \( \trmB \) must take the form \( \appl{\funA}{\vect{\trmB}{1}{k}{\cdots}} \) where
  \( \funA \) is a defined symbol and
  \( \trmB_i \) is \( \setReCa \)-computable for all \( i \).
  Let the type of \( \funA \) be denoted by
  \( \vect{\typA}{1}{m}{\arrType \cdots \arrType} \arrType \typB \)
  where \( \typB \) is a sort.
  Because \( \trmB = \appl{\funA}{\vect{\trmB}{1}{k}{\cdots}} \)
  is not \( \setReCa \)-computable,
  there exist \( \setReCa \)-computable terms \( \vect{\trmB}{k + 1}{m}{, \ldots,} \) such that
  \( \appl{\trmB}{\vect{\trmB}{k + 1}{m}{\cdots}} = \appl{\funA}{\vect{\trmB}{1}{m}{\cdots}} \)
  is not \( \setReCa \)-computable.

  Now due to \cref{cor:chainExistence}
  (with \( \setRule_0 = \setRule \) and \( \setRule_1 = \emptyset \)),
  there exists an infinite computable \( (\setStDP{\setRule}, \setRule) \)-chain.
  The non-existence of such a chain implies the termination of \( \setRule \).
\end{proof}

\subsection{The Proof of Theorem~\ref{thm:dpframeworkSoundness}}\label{app:proofDpframeworkSoundness}
We prove the soundness of each class of the DP processors separately.

\begin{theorem}
  Graph processors are sound.
\end{theorem}

\begin{proof}
  Given a DP problem \( \setDePa \) and
  a graph approximation \( (\grhA_\homA, \homA) \) for \( \setDePa \),
  every \( (\setDePa, \setRule) \)-chain induces a walk in \( \grhA_\homA \).
  In particular, assume given an infinite \( (\setDePa, \setRule) \)-chain
  \( (\sdpA_0, \subA_0), (\sdpA_1, \subA_1), \ldots \),
  then \( \homA(\sdpA_0), \homA(\sdpA_1), \ldots \) is an infinite sequence of vertices of \( \grhA_\homA \)
  such that there is an edge from \( \homA(\sdpA_{i - 1}) \) to \( \homA(\sdpA_i) \)
  for all \( i > 0 \).
  Since \( \grhA_\homA \) is finite, there exists \( N \) such that
  \( \bind{\forall}{i > N}{\bind{\exists}{k > i}{\homA(\sdpA_i) = \homA(\sdpA_k)}} \).
  Hence, there is a strongly connected component
  to which \( \homA(\sdpA_i) \) belong for all \( i > N \).
\end{proof}

The following result is not part of \cref{thm:dpframeworkSoundness},
and is not formally stated above;
we nevertheless need it for the correctness of the implementation.

\begin{lemma}
  The \textsc{Cap}-based approach in \cref{subsec:dpgraph}
  produces a graph approximation.
\end{lemma}

\begin{proof}
  Take SDPs \( \dpcv{\shrp{\trmB_0}}{\shrp{\trmA_0}}{\rulC_0}{\gtvA_0} \) and
  \( \dpcv{\shrp{\trmB_1}}{\shrp{\trmA_1}}{\rulC_1}{\gtvA_1} \),
  where variables are renamed if necessary to avoid name collisions.
  There ought to be an edge
  from \( \homA(\dpcv{\shrp{\trmB_0}}{\shrp{\trmA_0}}{\rulC_0}{\gtvA_0}) \)
  to \( \homA(\dpcv{\shrp{\trmB_1}}{\shrp{\trmA_1}}{\rulC_1}{\gtvA_1}) \)
  if there exists a substitution \( \subA \) such that
  \begin{enumerate*}
  \item \( \subA(\varA) \) is a ground theory term
    for all \( \varA \in \gtvA_0 \cup \gtvA_1 \),
  \item \( \mapNtrp{\rulC_0 \subA} = \mapNtrp{\rulC_1 \subA} = \algT \), and
  \item \( \shrp{\trmA_0} \subA \arrRwrt_\setRule^* \shrp{\trmB_1} \subA \).
  \end{enumerate*}
  Given such a substitution \( \subA \),
  the goal is to prove that
  \( \rulC_0 \wedge \rulC_1 \wedge \tcap(\shrp{\trmA_0}, \shrp{\trmB_1}) \) is satisfiable.
  We claim that \( \varA \mapsto \mapNtrp{\subA(\varA)} \) is a satisfying assignment.
  Since \( \mapNtrp{\rulC_0 \subA} = \mapNtrp{\rulC_1 \subA} = \algT \),
  we need only to show that \( \tcap(\shrp{\trmA_0}, \shrp{\trmB_1}) \)
  is satisfied by the assignment.
  To do so, we prove by induction on \( \trmC \)
  that \( \tcap(\trmC, \trmD) \) is satisfied for all \( \trmC \) and \( \trmD \) such that
  \( \trmC \subA \arrRwrt_\setRule^* \trmD \subA \).
  Consider each case of the definition of \( \tcap(\trmC, \trmD) \):
  \begin{itemize}
  \item \( \tcap(\trmC, \trmD) = \tcap(\trmC_1, \trmD_1) \wedge \cdots \wedge \tcap(\trmC_n, \trmD_n) \).
    In this case, \( \trmC_i \subA \arrRwrt_\setRule^* \trmD_i \subA \) for all \( i \).
    By induction, \( \tcap(\trmC_i, \trmD_i) \) is satisfied for all \( i \).
  \item \( \tcap(\trmC, \trmD) = \booF \).
    This case is impossible because
    its condition is inconsistent with \( \trmC \subA \arrRwrt_\setRule^* \trmD \subA \).
  \item \( \tcap(\trmC, \trmD) = (\trmC \relThEq \trmD) \).
    In this case, \( \trmC \subA \arrRwrt_\lblCalc^* \trmD \subA \).
    Hence, \( \subA(\varA) \) is a ground theory term for all \( \varA \in \setFvar{\trmD} \),
    and \( \mapNtrp{\trmC \subA} = \mapNtrp{\trmD \subA} \).
  \item \( \tcap(\trmC, \trmD) = \booT \).
    In this case, \( \tcap(\trmC, \trmD) \) is trivially satisfied.
    \qedhere
  \end{itemize}
\end{proof}

\begin{remark}
  This approach works for universal DP problems as well
  since function symbols in the base system cannot be defined in an extension.
\end{remark}

\begin{theorem}
  Subterm criterion processors are sound.
\end{theorem}

\begin{proof}
  Given a DP problem \( \setDePa \),
  a projection \( \prjA \) which \( \relSupp \)-orients
  a subset \( \setDePa^\prime \) of \( \setDePa \) and
  an infinite computable \( (\setDePa, \setRule) \)-chain
  \( (\dpcv{\shrp{\trmB_0}}{\shrp{\trmA_0}}{\rulC_0}{\gtvA_0}, \subA_0), (\dpcv{\shrp{\trmB_1}}{\shrp{\trmA_1}}{\rulC_1}{\gtvA_1}, \subA_1), \ldots \),
  by definition, for all \( i > 0 \), there exists a context \( \plug{\ctxA_i}{} \) such that
  \( \bar{\prjA}(\shrp{\trmA_{i - 1}} \subA_{i - 1}) \arrRwrt_\setRule^* \bar{\prjA}(\shrp{\trmB_i} \subA_i) = \plug{\ctxA_i}{\bar{\prjA}(\shrp{\trmA_i} \subA_i)} \).
  Because the \( (\setDePa, \setRule) \)-chain is computable,
  in particular, \( \bar{\prjA}(\shrp{\trmA_0} \subA_0) \)
  is \( \setReCa \)-computable and therefore terminating.
  Hence, there exists \( N \) such that for all \( i > N \),
  \( \bar{\prjA}(\shrp{\trmA_{i - 1}} \subA_{i - 1}) = \bar{\prjA}(\shrp{\trmB_i} \subA_i) \relSupt \bar{\prjA}(\shrp{\trmA_i} \subA_i) \).
  Now since \( \relSupp \) is well-founded,
  there exists \( N^\prime \) such that for all \( i > N^\prime \),
  \( \bar{\prjA}(\shrp{\trmB_i} \subA_i) = \bar{\prjA}(\shrp{\trmA_i} \subA_i) \).
  So \( \dpcv{\shrp{\trmB_i}}{\shrp{\trmA_i}}{\rulC_i}{\gtvA_i} \notin \setDePa^\prime \)
  for all \( i > N^\prime \).
\end{proof}

\begin{theorem}
  Integer mapping processors are sound.
\end{theorem}

\begin{proof}
  Given a DP problem \( \setDePa \),
  an integer mapping \( \inmA \) which \( > \)-orients
  a subset \( \setDePa^\prime \) of \( \setDePa \) and
  an infinite computable \( (\setDePa, \setRule) \)-chain
  \( (\dpcv{\shrp{\trmB_0}}{\shrp{\trmA_0}}{\rulC_0}{\gtvA_0}, \subA_0), (\dpcv{\shrp{\trmB_1}}{\shrp{\trmA_1}}{\rulC_1}{\gtvA_1}, \subA_1), \ldots \),
  by definition, for all \( i \),
  \( \bar{\inmA}(\shrp{\trmB_i}) \subA_i \) and \( \bar{\inmA}(\shrp{\trmA_i}) \subA_i \)
  are ground theory terms, and
  \( \bar{\inmA}(\shrp{\trmA_{i - 1}}) \subA_{i - 1} =
     \bar{\inmA}(\shrp{\trmA_{i - 1}} \subA_{i - 1}) \arrRwrt_\lblCalc^*
     \bar{\inmA}(\shrp{\trmB_i} \subA_i) =
     \bar{\inmA}(\shrp{\trmB_i}) \subA_i \)
  if \( i > 0 \).
  Let \( \subA^{\arrNorm_\lblCalc} \) be the substitution such that
  \( \subA^{\arrNorm_\lblCalc}(\varA) = \mapNorm{\subA(\varA)}_\lblCalc \),
  given any substitution \( \subA \).
  Then \( \subA_i^{\arrNorm_\lblCalc}(\varA) \) is a value
  for all \( i \) and \( \varA \in \gtvA_i \).
  Now for all \( i \),
  since \( \mapNtrp{\rulC_i \subA_i^{\arrNorm_\lblCalc}} = \mapNtrp{\rulC_i \subA_i} = \algT \), we have
  \( \mapNtrp{\bar{\inmA}(\shrp{\trmB_i}) \subA_i} =
     \mapNtrp{\bar{\inmA}(\shrp{\trmB_i}) \subA_i^{\arrNorm_\lblCalc}} \ge
     \mapNtrp{\bar{\inmA}(\shrp{\trmA_i}) \subA_i^{\arrNorm_\lblCalc}} =
     \mapNtrp{\bar{\inmA}(\shrp{\trmA_i}) \subA_i} \),
  and if \( \dpcv{\shrp{\trmB_i}}{\shrp{\trmA_i}}{\rulC_i}{\gtvA_i} \in \setDePa^\prime \),
  the inequality is strict with \( \mapNtrp{\bar{\inmA}(\shrp{\trmB_i}) \subA_i} \ge 0 \).
  Hence, there exists \( N \) such that for all \( i > N \),
  \( \dpcv{\shrp{\trmB_i}}{\shrp{\trmA_i}}{\rulC_i}{\gtvA_i} \notin \setDePa^\prime \).
\end{proof}

\begin{theorem}\label{thm:theoryargSoundness}
  Theory argument processors are sound.
\end{theorem}

\begin{proof}
  Given a DP problem \( \setDePa \),
  a theory argument mapping \( \tamA \) which fixes
  a subset \( \setDePa^\prime \) of \( \setDePa \) and
  an infinite computable \( (\setDePa, \setRule) \)-chain
  \( (\dpcv{\shrp{\trmB_0}}{\shrp{\trmA_0}}{\rulC_0}{\gtvA_0}, \subA_0), (\dpcv{\shrp{\trmB_1}}{\shrp{\trmA_1}}{\rulC_1}{\gtvA_1}, \subA_1), \ldots \),
  we consider two cases:
  \begin{itemize}
  \item If \( \dpcv{\shrp{\trmB_k}}{\shrp{\trmA_k}}{\rulC_k}{\gtvA_k} \in \setDePa^\prime \)
    for some \( k \) and
    \( \shrp{\trmA_k} = \appl{\shrp{\funA}}{\vect{\trmC}{1}{m}{\cdots}} \),
    by definition, for all \( i \in \tamA(\shrp{\funA}) \),
    \( \setFvar{\trmC_i} \subseteq \gtvA_k \) and therefore
    \( \trmC_i \subA_k \) is a ground theory term.
    Consider \( \shrp{\trmB_{k + 1}} = \appl{\shrp{\funA}}{\vect{\trmC^\prime}{1}{m}{\cdots}} \).
    For all \( i \), we have
    \( \trmC_i \subA_k \arrRwrt_\setRule^* \trmC_i^\prime \subA_{k + 1} \),
    and if \( i \in \tamA(\shrp{\funA}) \),
    \( \trmC_i^\prime \subA_{k + 1} \) is a ground theory term since \( \trmC_i \subA_k \) is.
    Hence, \( \subA_{k + 1}(\varA) \) is a ground theory term for all
    \( \varA \in \gtvA_{k + 1} \cup \bigcup_{i \in \tamA(\shrp{\funA})} \setFvar{\trmC_i^\prime} \),
    and \( \subA_{k + 1} \) respects
    \( \bar{\tamA}(\dpcv{\shrp{\trmB_{k + 1}}}{\shrp{\trmA_{k + 1}}}{\rulC_{k + 1}}{\gtvA_{k + 1}}) \).
    Now consider \( \shrp{\trmA_{k + 1}} = \appl{\shrp{\funB}}{\vect{\trmD}{1}{n}{\cdots}} \).
    It follows from the definition of a theory argument mapping that
    \( \trmD_i \subA_{k + 1} \) is a ground theory term for all \( i \in \tamA(\shrp{\funB}) \).
    We can thus repeat the above reasoning and conclude that
    \( \subA_i \) respects
    \( \bar{\tamA}(\dpcv{\shrp{\trmB_i}}{\shrp{\trmA_i}}{\rulC_i}{\gtvA_i}) \)
    for all \( i > k \).
    This gives us an infinite computable
    \( (\setComp{\bar{\tamA}(\sdpA)}{\sdpA \in \setDePa}, \setRule) \)-chain.
  \item If \( \dpcv{\shrp{\trmB_i}}{\shrp{\trmA_i}}{\rulC_i}{\gtvA_i} \notin \setDePa^\prime \)
    for any \( i \),
    the given \( (\setDePa, \setRule) \)-chain is in fact
    a \( (\setDePa \setminus \setDePa^\prime, \setRule) \)-chain.
    \qedhere
  \end{itemize}
\end{proof}

In order to prove the soundness of reduction pair processors,
we first present a lemma.

\begin{lemma}\label{lem:reductionpair}
  Given a constrained relation \( \succeq \) such that
  \( \crel{\succeq}{\rulL}{\rulR}{\rulC}{\setFvar{\rulC} \cup (\setFvar{\rulR} \setminus \setFvar{\rulL})} \)
  for all \( \rwrl{\rulL}{\rulR}{\rulC} \in \setRule \),
  and a monotonic relation \( \sqsupseteq \) which
  includes \( \arrRwrt_\lblCalc \) and
  covers \( \succeq \),
  for all terms \( \trmA \) and \( \trmA^\prime \)
  such that \( \trmA \arrRwrt_\setRule \trmA^\prime \),
  \( \mapNorm{\trmA}_\lblCalc \sqsupseteq^* \mapNorm{\trmA^\prime}_\lblCalc \).
\end{lemma}

\begin{proof}
  If \( \trmA \arrRwrt_\lblCalc \trmA^\prime \), we have
  \( \mapNorm{\trmA}_\lblCalc = \mapNorm{\trmA^\prime}_\lblCalc \).
  Now consider the case where
  \( \trmA = \plug{\ctxA}{\rulL \subA} \) and
  \( \trmA^\prime = \plug{\ctxA}{\rulR \subA} \) for
  some context \( \plug{\ctxA}{} \),
  some rewrite rule \( \rwrl{\rulL}{\rulR}{\rulC} \in \setRule \) and
  some substitution \( \subA \) which respects \( \rwrl{\rulL}{\rulR}{\rulC} \).
  By definition,
  \( \mapNorm{(\rulL \subA)}_\lblCalc \sqsupseteq \mapNorm{(\rulR \subA)}_\lblCalc \).
  Since \( \rulL \) is not a theory term,
  there exists a context \( \plug{\ctxA^\prime}{} \) such that
  \( \mapNorm{\trmA}_\lblCalc =
     \mapNorm{\plug{\ctxA}{\rulL \subA}}_\lblCalc =
     \plug{\ctxA^\prime}{\mapNorm{(\rulL \subA)}_\lblCalc} \).
  On the other hand,
  \( \mapNorm{\trmA^\prime}_\lblCalc =
     \mapNorm{\plug{\ctxA}{\rulR \subA}}_\lblCalc =
     \mapNorm{\plug{\ctxA^\prime}{\mapNorm{(\rulR \subA)}_\lblCalc}}_\lblCalc \).
  Due to the monotonicity of \( \sqsupseteq \),
  \( \plug{\ctxA^\prime}{\mapNorm{(\rulL \subA)}_\lblCalc} \sqsupseteq
     \plug{\ctxA^\prime}{\mapNorm{(\rulR \subA)}_\lblCalc} \arrRwrt_\lblCalc^*
     \mapNorm{\plug{\ctxA^\prime}{\mapNorm{(\rulR \subA)}_\lblCalc}}_\lblCalc \).
  Because \( {\arrRwrt_\lblCalc} \subseteq {\sqsupseteq} \),
  we have \( \mapNorm{\trmA}_\lblCalc \sqsupseteq^+ \mapNorm{\trmA^\prime}_\lblCalc \).
\end{proof}

\begin{theorem}
  Reduction pair processors are sound.
\end{theorem}

\begin{proof}
  Given a DP problem \( \setDePa \),
  a constrained reduction pair \( (\succeq, \succ) \)
  (covered by \( (\sqsupseteq, \sqsupset) \))
  with the said properties for a subset \( \setDePa^\prime \) of \( \setDePa \) and
  an infinite computable \( (\setDePa, \setRule) \)-chain
  \( (\dpcv{\shrp{\trmB_0}}{\shrp{\trmA_0}}{\rulC_0}{\gtvA_0}, \subA_0), (\dpcv{\shrp{\trmB_1}}{\shrp{\trmA_1}}{\rulC_1}{\gtvA_1}, \subA_1), \ldots \),
  by definition, for all \( i \),
  \( \mapNorm{(\shrp{\trmB_i} \subA_i)}_\lblCalc \sqsupset \mapNorm{(\shrp{\trmA_i} \subA_i)}_\lblCalc \)
  if \( \dpcv{\shrp{\trmB_i}}{\shrp{\trmA_i}}{\rulC_i}{\gtvA_i} \in \setDePa^\prime \), and
  \( \mapNorm{(\shrp{\trmB_i} \subA_i)}_\lblCalc \sqsupseteq \mapNorm{(\shrp{\trmA_i} \subA_i)}_\lblCalc \)
  if \( \dpcv{\shrp{\trmB_i}}{\shrp{\trmA_i}}{\rulC_i}{\gtvA_i} \notin \setDePa^\prime \).
  Due to \cref{lem:reductionpair},
  \( \mapNorm{\shrp{(\trmA_{i - 1}} \subA_{i - 1})}_\lblCalc \sqsupseteq^* \mapNorm{(\shrp{\trmB_i} \subA_i)}_\lblCalc \)
  for all \( i > 0 \).
  Now since \( \sqsupset \) is well-founded
  and \( \relComp{\sqsupseteq}{\sqsupset} \subseteq {\sqsupset^+} \),
  there exists \( N \) such that for all \( i > N \),
  \( \dpcv{\shrp{\trmB_i}}{\shrp{\trmA_i}}{\rulC_i}{\gtvA_i} \notin \setDePa^\prime \).
\end{proof}

\subsection{The Proof of Theorem~\ref{thm:chainComputability}}\label{app:proofChainComputability}
For the properties of \( \setReCa \)-computability,
see Appendix A in the extended version\footnote{\url{https://doi.org/10.48550/arXiv.1902.06733}} of \cite{fuh:kop:19}.
In order to prove \cref{thm:chainComputability},
we first present two lemmas.

\begin{lemma}
  Undefined function symbols are \( \setReCa \)-computable.
\end{lemma}

\begin{proof}
  Given an LCSTRS \( \setRule \) and an undefined function symbol
  \( \funA \relType \vect{\typA}{1}{n}{\arrType \cdots \arrType} \arrType \typB \)
  where \( \typB \) is a sort,
  take arbitrary \( \setReCa \)-computable terms \( \vect{\trmA}{1}{n}{, \ldots,} \)
  such that \( \trmA_i \relType \typA_i \) for all \( i \).
  Consider the reducts of \( \appl{\funA}{\vect{\trmA}{1}{n}{\cdots}} \),
  each of which must be either
  \( \appl{\funA}{\vect{\trmA^\prime}{1}{n}{\cdots}} \)
  where \( \trmA_i \arrRwrt_\setRule \trmA_i^\prime \) for all \( i \) or
  a value (when \( \funA \) is a theory symbol)
  because \( \funA \) is undefined.
  By definition, \( \appl{\funA}{\vect{\trmA}{1}{n}{\cdots}} \in \setReCa \),
  and therefore \( \funA \) is \( \setReCa \)-computable.
\end{proof}

\begin{lemma}\label{lem:chainStep}
  Given an AFP system \( \setRule_0 \) with sort ordering \( \relSord \),
  an extension \( \setRule_1 \) of \( \setRule_0 \) and
  a sort ordering \( \relSord^\prime \)
  which extends \( \relSord \) over sorts in \( \setRule_0 \cup \setRule_1 \),
  for each defined symbol
  \( \funA \relType \vect{\typA}{1}{m}{\arrType \cdots \arrType} \arrType \typB \)
  in \( \setRule_0 \) where \( \typB \) is a sort,
  if \( \appl{\funA}{\vect{\trmB}{1}{m}{\cdots}} \)
  is not \( \setReCa \)-computable
  in \( \setRule_0 \cup \setRule_1 \) with \( \relSord^\prime \)
  but \( \trmB_i \) is for all \( i \),
  there exist an SDP
  \( \dpcv{\appl{\shrp{\funA}}{\vect{\trmB^\prime}{1}{m}{\cdots}}}{\appl{\shrp{\funB}}{\vect{\trmA}{1}{n}{\cdots}}}{\rulC}{\gtvA} \in \setStDP{\setRule_0} \)
  and a substitution \( \subA \) such that
  \begin{enumerate*}
  \item \( \trmB_i \arrRwrt_{\setRule_0 \cup \setRule_1}^* \trmB_i^\prime \subA \) for all \( i \),
  \item \( \subA \) respects
    \( \dpcv{\appl{\shrp{\funA}}{\vect{\trmB^\prime}{1}{m}{\cdots}}}{\appl{\shrp{\funB}}{\vect{\trmA}{1}{n}{\cdots}}}{\rulC}{\gtvA} \), and
  \item \( (\appl{\funB}{\vect{\trmA}{1}{n}{\cdots}}) \subA = \appl{\funB}{(\trmA_1 \subA) \cdots (\trmA_n \subA)} \)
    is not \( \setReCa \)-computable
    in \( \setRule_0 \cup \setRule_1 \) with \( \relSord^\prime \)
    but \( \trmC \subA \) is
    for each proper subterm \( \trmC \) of \( \appl{\funB}{\vect{\trmA}{1}{n}{\cdots}} \).
  \end{enumerate*}
\end{lemma}

\begin{proof}
  We consider \( \setReCa \)-computability in
  \( \setRule_0 \cup \setRule_1 \) with \( \relSord^\prime \).
  If all the reducts of \( \appl{\funA}{\vect{\trmB}{1}{m}{\cdots}} \) were either
  \( \appl{\funA}{\vect{\trmB^\prime}{1}{m}{\cdots}} \)
  where \( \trmB_i \arrRwrt_{\setRule_0 \cup \setRule_1}^* \trmB_i^\prime \) for all \( i \) or
  a value (when \( \funA \) is a theory symbol),
  \( \appl{\funA}{\vect{\trmB}{1}{m}{\cdots}} \) would be computable.
  Hence, there exist a rewrite rule
  \( \rwrl{\appl{\funA}{\vect{\trmB^\prime}{1}{p}{\cdots}}}{\rulR}{\rulC} \in \setRule_0 \)
  (\( \funA \) cannot be defined in \( \setRule_1 \)) and
  a substitution \( \subA^\prime \) such that
  \( \trmB_i \arrRwrt_{\setRule_0 \cup \setRule_1}^* \trmB_i^\prime \subA^\prime \)
  for all \( i \le p \) and
  \( \subA^\prime \) respects the rewrite rule;
  \( \appl{(\rulR \subA^\prime)}{\vect{\trmB}{p + 1}{m}{\cdots}} \)
  is thus a reduct of \( \appl{\funA}{\vect{\trmB}{1}{m}{\cdots}} \).
  There is at least one such reduct that is uncomputable---otherwise,
  \( \appl{\funA}{\vect{\trmB}{1}{m}{\cdots}} \) would be computable.
  Let \( \appl{(\rulR \subA^\prime)}{\vect{\trmB}{p + 1}{m}{\cdots}} \) be uncomputable,
  and therefore so is \( \rulR \subA^\prime \).

  Take a minimal subterm \( \trmA \) of \( \rulR \)
  such that \( \trmA \subA^\prime \) is uncomputable.
  If \( \trmA = \appl{\varA}{\vect{\trmA}{1}{q}{\cdots}} \) for some variable \( \varA \),
  \( \subA^\prime(\varA) \) is either a value or
  an accessible subterm of \( \trmB_k^\prime \subA^\prime \) for some \( k \)
  because \( \setRule_0 \) is AFP.
  Either way, \( \subA^\prime(\varA) \) is computable.
  Due to the minimality of \( \trmA \),
  \( \trmA_i \subA^\prime \) is computable for all \( i \),
  which implies the computability of
  \( \trmA \subA^\prime = \appl{\subA^\prime(\varA)}{(\trmA_1 \subA^\prime) \cdots (\trmA_q \subA^\prime)} \).
  This contradiction shows that \( \trmA = \appl{\funB}{\vect{\trmA}{1}{q}{\cdots}} \)
  for some function symbol \( \funB \) in \( \setRule_0 \).
  And \( \funB \) must be a defined symbol.

  Now we have an SDP
  \( \dpcv{\appB{\shrp{\funA}}{\vect{\trmB^\prime}{1}{p}{\cdots}}{\vect{\varA}{p + 1}{m}{\cdots}}}{\appB{\shrp{\funB}}{\vect{\trmA}{1}{q}{\cdots}}{\vect{\varB}{q + 1}{n}{\cdots}}}{\rulC}{\setFvar{\rulC} \cup (\setFvar{\rulR} \setminus \setFvar{\appl{\funA}{\vect{\trmB^\prime}{1}{p}{\cdots}}})} \in \setStDP{\setRule_0} \).
  Because \( \trmA \subA^\prime \) is uncomputable,
  there exist computable terms \( \vect{\trmA^\prime}{q + 1}{n}{, \ldots,} \) such that
  \( \appl{(\trmA \subA^\prime)}{\vect{\trmA^\prime}{q + 1}{n}{\cdots}} = \appB{\funB}{(\trmA_1 \subA^\prime) \cdots (\trmA_q \subA^\prime)}{\vect{\trmA^\prime}{q + 1}{n}{\cdots}} \)
  is uncomputable.
  Let \( \subA \) be the substitution such that
  \( \subA(\varA_i) = \trmB_i \) for all \( i > p \),
  \( \subA(\varB_i) = \trmA_i^\prime \) for all \( i > q \),
  and \( \subA(\varC) = \subA^\prime(\varC) \) for any other variable \( \varC \).
  Let \( \trmB_i^\prime \) denote \( \varA_i \) for all \( i > p \),
  let \( \trmA_i \) denote \( \varB_i \) for all \( i > q \),
  and let \( \gtvA \) denote
  \( \setFvar{\rulC} \cup (\setFvar{\rulR} \setminus \setFvar{\appl{\funA}{\vect{\trmB^\prime}{1}{p}{\cdots}}}) \),
  then \( \dpcv{\appl{\shrp{\funA}}{\vect{\trmB^\prime}{1}{m}{\cdots}}}{\appl{\shrp{\funB}}{\vect{\trmA}{1}{n}{\cdots}}}{\rulC}{\gtvA} \) and
  \( \subA \) satisfy all the requirements.
\end{proof}

\begin{corollary}\label{cor:chainExistence}
  Given an AFP system \( \setRule_0 \) with sort ordering \( \relSord \),
  an extension \( \setRule_1 \) of \( \setRule_0 \) and
  a sort ordering \( \relSord^\prime \)
  which extends \( \relSord \) over sorts in \( \setRule_0 \cup \setRule_1 \),
  for each defined symbol
  \( \funA \relType \vect{\typA}{1}{m}{\arrType \cdots \arrType} \arrType \typB \)
  in \( \setRule_0 \) where \( \typB \) is a sort,
  if \( \appl{\funA}{\vect{\trmB}{1}{m}{\cdots}} \)
  is not \( \setReCa \)-computable
  in \( \setRule_0 \cup \setRule_1 \) with \( \relSord^\prime \)
  but \( \trmB_i \) is for all \( i \),
  there exists an infinite computable
  \( (\setStDP{\setRule_0}, \setRule_0 \cup \setRule_1) \)-chain
  \( (\dpcv{\appl{\shrp{\funA}}{\vect{\trmB^\prime}{1}{m}{\cdots}}}{\shrp{\trmA}}{\rulC}{\gtvA}, \subA), \ldots \)
\end{corollary}

\begin{proof}
  Repeatedly applying \cref{lem:chainStep}---first on
  \( \appl{\funA}{\vect{\trmB}{1}{m}{\cdots}} \),
  then on \( \appl{\funB}{(\trmA_1 \subA) \cdots (\trmA_n \subA)} \) and
  so on---we thus get such an infinite computable
  \( (\setStDP{\setRule_0}, \setRule_0 \cup \setRule_1) \)-chain.
\end{proof}

\chainComputability*

\begin{proof}
  Assume that \( \setRule_0 \) is not publicly computable.
  By definition, there exist
  a public extension \( \setRule_1 \) of \( \setRule_0 \),
  a sort ordering \( \relSord^\prime \)
  which extends \( \relSord \) over sorts in \( \setRule_0 \cup \setRule_1 \) and
  a public term \( \trmC \) of \( \setRule_0 \)
  such that \( \trmC \) is not \( \setReCa \)-computable
  in \( \setRule_0 \cup \setRule_1 \) with \( \relSord^\prime \).
  We consider \( \setReCa \)-computability in
  \( \setRule_0 \cup \setRule_1 \) with \( \relSord^\prime \).
  Take a minimal uncomputable subterm \( \trmB \) of \( \trmC \),
  then \( \trmB \) must take the form \( \appl{\funA}{\vect{\trmB}{1}{k}{\cdots}} \) where
  \( \funA \) is a defined symbol in \( \setRule_0 \) and
  \( \trmB_i \) is computable for all \( i \).
  Let the type of \( \funA \) be denoted by
  \( \vect{\typA}{1}{m}{\arrType \cdots \arrType} \arrType \typB \)
  where \( \typB \) is a sort.
  Because \( \trmB = \appl{\funA}{\vect{\trmB}{1}{k}{\cdots}} \) is uncomputable,
  there exist computable terms
  \( \vect{\trmB}{k + 1}{m}{, \ldots,} \) of \( \setRule_0 \cup \setRule_1 \) such that
  \( \appl{\trmB}{\vect{\trmB}{k + 1}{m}{\cdots}} = \appl{\funA}{\vect{\trmB}{1}{m}{\cdots}} \)
  is uncomputable.

  Now due to \cref{cor:chainExistence},
  there exists an infinite computable
  \( (\setStDP{\setRule_0}, \setRule_0 \cup \setRule_1) \)-chain
  \( (\dpcv{\appl{\shrp{\funA}}{\vect{\trmB^\prime}{1}{m}{\cdots}}}{\shrp{\trmA}}{\rulC}{\gtvA}, \subA), \ldots \),
  which is public because \( \funA \) is not a hidden symbol by construction.
  The non-existence of such a chain implies the public computability of \( \setRule_0 \).
\end{proof}

\subsection{The Proof of Theorem~\ref{thm:computabilitySoundness}}\label{app:proofComputabilitySoundness}
We prove the soundness of each class of the DP processors separately.

\begin{theorem}
  Theory argument processors are sound when the input flag is \( \flagPub \).
\end{theorem}

\begin{proof}
  We refer to the proof of \cref{thm:theoryargSoundness},
  and consider the two cases in the definition of a theory argument processor
  when the input is \( (\setDePa, \flagPub) \):
  \begin{itemize}
  \item If \( \setDePa^\prime \) contains all the public SDPs in \( \setDePa \),
    the first SDP in a public \( (\setDePa, \setRule_0 \cup \setRule_1) \)-chain must be an element of \( \setDePa^\prime \).
    Every infinite computable \( (\setDePa, \setRule_0 \cup \setRule_1) \)-chain which is public
    thus gives rise to an infinite computable
    \( (\setComp{\sdpA}{\sdpA \in \setDePa\text{ is public}} \cup \setComp{\bar{\tamA}(\sdpA)}{\sdpA \in \setDePa\text{ is not public}}, \setRule_0 \cup \setRule_1) \)-chain
    which is public.
  \item If \( \setDePa \setminus \setDePa^\prime \) contains at least one public SDP,
    every infinite computable \( (\setDePa, \setRule_0 \cup \setRule_1) \)-chain which is public
    either gives rise to an infinite computable
    \( (\setComp{\bar{\tamA}(\sdpA)}{\sdpA \in \setDePa}, \setRule_0 \cup \setRule_1) \)-chain
    which may or may not be public,
    or is a \( (\setDePa \setminus \setDePa^\prime, \setRule_0 \cup \setRule_1) \)-chain.
    \qedhere
  \end{itemize}
\end{proof}

\begin{theorem}
  Constraint modification processors are sound.
\end{theorem}

\begin{proof}
  Given a DP problem \( (\setDePa, \pflg) \),
  a set \( \setDePa^\prime \) of SDPs which covers \( \setDePa \) and
  an infinite computable \( (\setDePa, \setRule_0 \cup \setRule_1) \)-chain
  \( (\dpcv{\shrp{\trmB_0}}{\shrp{\trmA_0}}{\rulC_0}{\gtvA_0}, \subA_0), (\dpcv{\shrp{\trmB_1}}{\shrp{\trmA_1}}{\rulC_1}{\gtvA_1}, \subA_1), \ldots \),
  by definition, for all \( i \),
  there exist \( \dpcv{\shrp{\trmB_i}}{\shrp{\trmA_i}}{\rulC_i^\prime}{\gtvA_i^\prime} \in \setDePa^\prime \)
  and \( \subA_i^\prime \) such that
  \( \subA_i^\prime \) respects \( \dpcv{\shrp{\trmB_i}}{\shrp{\trmA_i}}{\rulC_i^\prime}{\gtvA_i^\prime} \),
  \( \trmB_i \subA_i = \trmB_i \subA_i^\prime \) and
  \( \trmA_i \subA_i = \trmA_i \subA_i^\prime \).
  Hence,
  \( (\dpcv{\shrp{\trmB_0}}{\shrp{\trmA_0}}{\rulC_0^\prime}{\gtvA_0^\prime}, \subA_0^\prime), (\dpcv{\shrp{\trmB_1}}{\shrp{\trmA_1}}{\rulC_1^\prime}{\gtvA_1^\prime}, \subA_1^\prime), \ldots \)
  is an infinite computable \( (\setDePa^\prime, \setRule_0 \cup \setRule_1) \)-chain,
  and is public if the given \( (\setDePa, \setRule_0 \cup \setRule_1) \)-chain is.
\end{proof}

\begin{theorem}
  Reachability processors are sound.
\end{theorem}

\begin{proof}
  Graph approximations over-approximate the DP graph,
  and SDPs that are unreachable from any public SDP
  cannot contribute to a public \( (\setDePa, \setRule_0 \cup \setRule_1) \)-chain.
\end{proof}

\end{document}